\documentclass[twoside,11pt]{article}
\usepackage{jmlr2e}
\RequirePackage{multirow,booktabs,subfigure,color,array,hhline,makecell}
\usepackage{amsmath}
\usepackage{amssymb}

\allowdisplaybreaks
\RequirePackage{algorithm,algpseudocode}
\newtheorem{thm}{Theorem}
\newtheorem{defin}[thm]{Definition}
\newtheorem{lem}[thm]{Lemma}
\newtheorem{assum}[thm]{Assumption}
\newtheorem{rem}[thm]{Remark}
\newtheorem{cor}[thm]{Corollary}
\newtheorem{prop}[thm]{Proposition}
\ShortHeadings{Overlapping and nonoverlapping models}{Qing}
\firstpageno{1}
\begin{document}
	\title{Overlapping and nonoverlapping models}
	\author{\name Huan Qing \email qinghuan@cumt.edu.cn \\
		\addr School of Mathematics\\
		China University of Mining and Technology\\
		Xuzhou, 221116, P.R. China}
	\editor{}
	\maketitle
\begin{abstract}
Consider a directed network with $K_{r}$ row communities and $K_{c}$ column communities. Previous works found that modeling directed networks in which all nodes have overlapping property requires $K_{r}=K_{c}$ for identifiability. In this paper, we propose an overlapping and nonoverlapping model to study directed networks in which row nodes have overlapping property while column nodes do not. The proposed model is identifiable when $K_{r}\leq K_{c}$. Meanwhile, we provide one identifiable model as extension of ONM to model directed networks with variation in node degree. Two spectral algorithms with theoretical guarantee on consistent estimations are designed to fit the models. A small scale of numerical studies are used to illustrate the algorithms.
\end{abstract}
\section{Introduction}\label{sec1}
In the study of social networks, various models have been proposed to learn the latent structure of networks. Due to the extremely intensive studies on community detection, we only focus on identifiable models that are closely relevant to our study in this paper. For undirected network, the Stochastic Blockmodel (SBM) \citep{SBM} is a classical and widely used model to generate undirected networks. The degree-corrected stochastic blockmodel (DCSBM) \cite{DCSBM} extends SBM by introducing degree heterogeneities. Under SBM and DCSBM, all nodes are pure such that each node only belong to one community.  While, in real cases some nodes may belong to multiple communities, and such nodes have overlapping (also known as mixed membership) property. To model undirected networks in which nodes have overlapping property, \cite{MMSB} designs the Mixed Membership Stochastic Blockmodel (MMSB). \cite{MixedSCORE} introduces the degree-corrected mixed membership model (DCMM) which extends MMSB by considering degree heterogeneities. \cite{OCCAM} designs the OCCAM model which equals DCMM actually. Spectral methods with consistent estimations under the above models are provided in \cite{rohe2011spectral,RSC,lei2015consistency,joseph2016impact,SCORE,MixedSCORE,mao2020estimating, MaoSVM}. For directed networks in which all nodes have nonoverlapping property, \cite{DISIM} proposes a model called Stochastic co-Blockmodel (ScBM) and its extension DCScBM by considering degree heterogeneity, where ScBM (DCScBM) is extension of SBM (DCSBM). ScBM and DCScBM can model nonoverlapping directed networks in which row nodes belong to $K_{r}$ row communities and column nodes belong to $K_{c}$ column communities, where $K_{r}$ can differ from $K_{c}$. \cite{zhou2019analysis, qing2021consistency} study the consistency of some adjacency-based spectral algorithms under ScBM. \cite{DSCORE} studies the consistency of the spectral method D-SCORE under DCScBM when $K_{r}=K_{c}$. \cite{qing2021DiMMSB} designs directed mixed membership stochastic blockmodel (DiMMSB) as an extension of ScBM and MMSB to model directed networks in which all nodes have overlapping property. Meanwhile, DiMMSB can also be seen as an extension of the two-way blockmodels with Bernoulli distribution of \cite{airoldi2013multi}.  All the above models are identifiable under certain conditions. The identifiability of ScBM and DCScBM holds even for the case when $K_{r}\neq K_{c}$. DiMMSB is identifiable only when $K_{r}=K_{c}$. Sure, SBM, DCSBM, MMSB, DCMM and OCCAM are identifiable when $K_{r}=K_{c}$ since they model undirected networks. For all the above models, row nodes and column nodes have symmetric structural information such that they always have nonoverlapping property or overlapping property simultaneously. As shown by the identifiability of DiMMSB, to model a directed network in which all nodes have overlapping property, the identifiability of the model requires $K_{r}=K_{c}$. Naturally, there is a bridge model from ScBM to DiMMSB such that the bride model can model a directed network in which row nodes and column nodes have asymmetric structural information such that they have different overlapping property. In this paper, we introduce this model and name it as overlapping and nonoverlapping model.

Our contributions in this paper are as follows. We propose an identifiable model for directed networks, the overlapping and nonoverlapping model (ONM for short). ONM allows that nodes in a directed network can have different overlapping property. Without loss of generality, in a directed network, we let row nodes have overlapping property while column nodes do not.  The proposed model is identifiable when $K_{r}\leq K_{c}$. Recall that the identifiability of ScBM modeling nonoverlapping directed networks holds even for the case $K_{r}\neq K_{c}$, and DiMMSB modeling overlapping directed networks is identifiable only when $K_{r}=K_{c}$, this is the reason we call ONM modeling directed networks in which row nodes have different overlapping property as column nodes as a bridge model from ScBM to DiMMSB. Similar as DCScBM is an extension of ScBM, we propose an identifiable model overlapping and degree-corrected nonoverlapping model (ODCNM) as extension of ONM by considering degree heterogeneity. We construct two spectral algorithms to fit ONM and ODCNM. We show that our method enjoy consistent estimations under mild conditions by delicate spectral analysis. Especially, our theoretical results under ODCNM match those under ONM when ODCNM degenerates to ONM.

\textbf{\textit{Notations.}}
We take the following general notations in this paper. For any positive integer $m$, let $[m]:=\{1,2,\ldots, m\}$. For a vector $x$ and fixed $q>0$, $\|x\|_{q}$ denotes its $l_{q}$-norm. For a matrix $M$, $M'$ denotes the transpose of the matrix $M$, $\|M\|$ denotes the spectral norm, $\|M\|_{F}$ denotes the Frobenius norm, and $\|M\|_{2\rightarrow\infty}$ denotes the maximum $l_{2}$-norm of all the rows of $M$. Let $\sigma_{i}(M)$ be the $i$-th largest singular value of matrix $M$, and $\lambda_{i}(M)$ denote the $i$-th largest eigenvalue of the matrix $M$ ordered by the magnitude. $M(i,:)$ and $M(:,j)$ denote the $i$-th row and the $j$-th column of matrix $M$, respectively. $M(S_{r},:)$ and $M(:,S_{c})$ denote the rows and columns in the index sets $S_{r}$ and $S_{c}$ of matrix $M$, respectively. For any matrix $M$, we simply use $Y=\mathrm{max}(0, M)$ to represent $Y_{ij}=\mathrm{max}(0, M_{ij})$ for any $i,j$.  For any matrix $M\in\mathbb{R}^{m\times m}$, let $\mathrm{diag}(M)$ be the $m\times m$ diagonal matrix whose $i$-th diagonal entry is $M(i,i)$. $\mathbf{1}$ is a column vector with all entries being ones. $e_{i}$ is a column vector whose $i$-th entry is 1 while other entries are zero.
\section{The overlapping and nonoverlapping model}\label{sec2}
Consider a directed network $\mathcal{N}=(V_{r}, V_{c},E)$, where $V_{r}=\{1,2,\ldots, n_{r}\}$ is the set of row nodes, $V_{c}=\{1,2,\ldots, n_{c}\}$ is the set of column nodes, and $E$ is the set of edges from row nodes to column nodes. Note that since row nodes can be different from column nodes, we may have $V_{r}\cap V_{c}=\varnothing$, where $\varnothing$ denotes the null set. In this paper, we use subscript $r$ and $c$ to distinguish terms for row nodes and column nodes. Let $A\in \{0,1\}^{n_{r}\times n_{c}}$ be the bi-adjacency matrix of directed network $\mathcal{N}$ such that $A(i_{r},i_{c})=1$ if there is a directional edge from row node $i_{r}$ to column node $i_{c}$, and $A(i_{r},i_{c})=0$ otherwise.

We propose a new block model which we call overlapping and nonoverlapping model (ONM for short). ONM can model directed networks whose row nodes belong to $K_{r}$ overlapping row communities while column nodes belong to $K_{c}$ nonoverlapping column communities.

For row nodes, let $\Pi_{r}\in \mathbb{R}^{n_{r}\times K_{r}}$ be the membership matrix of row nodes such that
\begin{align}\label{Pir}
 \Pi_{r}(i_{r},)\geq 0, \|\Pi_{r}(i_{r},:)\|_{1}=1\qquad \mathrm{for~} i_{r}\in[n_{r}].
\end{align}
Call row node $i_{r}$ \emph{pure} if $\Pi_{r}(i_{r},:)$ degenerates (i.e., one entry is 1, all others $K_{r}-1$ entries are 0) and \emph{mixed} otherwise. From such definition, row node $i_{r}$ has mixed membership and may belong to more than one row communities for $i_{r}\in[n_{r}]$.

For column nodes, let $\ell$ be the $n_{c}\times 1$ vector whose $i_{c}$-th entry $\ell(i_{c})=k$ if column node $i_{c}$ belongs to the $k$-th column community, and $\ell(i_{c})$ takes value from $\{1,2,\ldots, K_{c}\}$ for $i_{c}\in[n_{c}]$. Let $\Pi_{c}\in \mathbb{R}^{n_{c}\times K_{c}}$ be the membership matrix of column nodes such that
\begin{align}\label{Pic}
\Pi_{c}(i_{c},k)=1~\mathrm{when~}\ell(i_{c})=k,\mathrm{~and~}0\mathrm{~otherwise}, \mathrm{and~}\|\Pi_{c}(i_{c},:)\|_{1}=1  \qquad \mathrm{for~}i_{c}\in[n_{c}], k\in[K_{c}].
\end{align}
From such definition, column node $i_{c}$ belongs to \emph{exactly} one of the $K_{c}$ column communities for $i_{c}\in[n_{c}]$. Sure, all column nodes are pure nodes.

In this paper, we assume that
\begin{align}\label{krkc}
K_{r}\leq K_{c}.
\end{align}
Eq (\ref{krkc}) is required for the identifiability of ONM.

Let $P\in \mathbb{R}^{K_{r}\times K_{c}}$ be the probability matrix (also known as connectivity matrix) such that
\begin{align}\label{definP}
 0\leq P(k,l)\leq \rho\leq 1 \qquad \mathrm{for~} k\in[K_{r}], l\in[K_{c}],
\end{align}
where $\rho$ controls the network sparsity and is called sparsity parameter in this paper. For convenience, set $P=\rho \tilde{P}$ where $\tilde{P}(k,l)\in[0,1]$ for $k\in[K_{r}], l\in[K_{c}]$, and $\mathrm{max}_{k\in[K_{r}],l\in[K_{c}]}\tilde{P}(k,l)=1$ for model identifiability. For all pairs of $(i_{r},i_{c})$ with $i_{r}\in[n_{r}],i_{c}\in[n_{c}]$, our model assumes that $A(i_{r},i_{c})$ are independent Bernoulli random variables satisfying
\begin{align}\label{ONM}
\Omega:=\Pi_{r}P\Pi'_{c},~~~A(i_{r},i_{c})\sim\mathrm{Bernoulli}(\Omega(i_{r},i_{c}))\qquad \mathrm{for~}i_{r}\in[n_{r}],i_{c}\in[n_{c}],
\end{align}
where $\Omega=\mathbb{E}[A]$ , and we call it population adjacency matrix in this paper.
\begin{defin}
Call model (\ref{Pir})-(\ref{ONM}) the Overlapping and Nonoverlapping model (ONM) and denote it by $ONM_{n_{r},n_{c}}(K_{r},K_{c}, P, \Pi_{r}, \Pi_{c})$.
\end{defin}
The following conditions are sufficient for the identifiability of ONM:
\begin{itemize}
  \item (I1) $\mathrm{rank}(P)=K_{r}, \mathrm{rank}(\Pi_{r})=K_{r}$ and $\mathrm{rank}(\Pi_{c})=K_{c}$.
  \item (I2) There is at least one pure row node for each of the $K_{r}$ row communities.
\end{itemize}
Here, $\mathrm{rank}(\Pi_{r})=K_{r}$ means that $\sum_{i_{r}=1}^{n_{r}}(\Pi_{r}(i_{r},k))>0$ for all $k\in[K_{r}]$; $\mathrm{rank}(\Pi_{c})=K_{c}$ means that each column community has at least one column node. For $k\in[K_{r}]$, let $\mathcal{I}^{(k)}_{r}=\{i\in\{1,2,\ldots, n_{r}\}: \Pi_{r}(i,k)=1\}$. By condition (I2), $I^{(k)}_{r}$ is non empty for all $k\in[K_{r}]$. For $k\in[K_{r}]$, select one row node from $\mathcal{I}^{(k)}_{r}$ to construct the index set $\mathcal{I}_{r}$, i.e., $\mathcal{I}_{r}$ is the indices of row nodes corresponding to $K_{r}$ pure row nodes, one from each community. W.L.O.G., let $\Pi_{r}(\mathcal{I}_{r},:)=I_{K_{r}}$ (Lemma 2.1 \cite{mao2020estimating} also has similar setting to design their spectral algorithms under MMSB.). $\mathcal{I}_{c}$ is defined similarly for column nodes such that $\Pi_{c}(\mathcal{I}_{c},:)=I_{K_{c}}$. Next proposition guarantees that once conditions (I1) and (I2) hold, ONM is identifiable.
\begin{prop}\label{id}
	If conditions (I1) and (I2) hold, ONM is identifiable: For eligible $(P,\Pi_{r}, \Pi_{c})$ and $(\check{P},\check{\Pi}_{r}, \check{\Pi}_{c})$, if $\Pi_{r}P\Pi'_{c}=\check{\Pi}_{r}\check{P}\check{\Pi}'_{c}$, then  $P=\check{P}, \Pi_{r}=\check{\Pi}_{r}$, and $\Pi_{c}=\check{\Pi}_{c}$.
\end{prop}
Compared to some previous models for directed networks, ONM models different directed networks.
\begin{itemize}
\item When all row nodes are pure, our ONM reduces to ScBM with $K_{r}$ row clusters and $K_{c}$ column clusters  \cite{DISIM}. However, ONM allows row nodes to have overlapping memberships while ScBM does not. Meanwhile, for model identifiability, ScBM does not require $\mathrm{rank}(P)=K_{r}$ while ONM requires, and this can be seen as the cost of ONM when modeling overlapping row nodes.
\item Though DiMMSB \cite{qing2021DiMMSB} can model directed networks whose row and column nodes have overlapping memberships, DiMMSB requires $K_{r}=K_{c}$ for model identifiability. For comparison, our ONM allows $K_{r}\leq K_{c}$ at the cost of losing overlapping property of column nodes.
\end{itemize}
\subsection{A spectral algorithm for fitting ONM}
The primary goal of the proposed algorithm is to estimate the row membership matrix $\Pi_{r}$ and column membership matrix $\Pi_{c}$ from the observed adjacency matrix $A$ with given $K_{r}$ and $K_{c}$.

We now discuss our intuition for the design of our algorithm to fit ONM. Under conditions (I1) and (I2), by basic algebra, we have $\mathrm{rank}(\Omega)=K_{r}$. Let $\Omega=U_{r}\Lambda U'_{c}$ be the compact singular value decomposition of $\Omega$, where  $U_{r}\in\mathbb{R}^{n_{r}\times K_{r}}, \Lambda\in\mathbb{R}^{K_{r}\times K_{r}}, U_{c}\in\mathbb{R}^{n_{c}\times K_{r}}$, $U'_{r}U_{r}=I_{K_{r}}, U'_{c}U_{c}=I_{K_{r}}$, and $I_{K_{r}}$ is a $K_{r}\times K_{r}$ identity matrix. Let $n_{c,k}=|\{i_{c}:\ell(i_{c})=k\}|$ be the size of the $k$-th column community for $k\in[K_{c}]$. Let $n_{c,\mathrm{max}}=\mathrm{max}_{k\in[K_{c}]}n_{c,k}$ and $n_{c,\mathrm{min}}=\mathrm{min}_{k\in[K_{c}]}n_{c,k}$. Meanwhile, without causing confusion, let $n_{c,K_{r}}$ be the $K_{r}$-th largest size among all column communities. The following lemma guarantees that $U_{r}$ enjoys ideal simplex structure and $U_{c}$ has $K_{c}$ distinct rows.
\begin{lem}\label{RK}
Under $ONM_{n_{r},n_{c}}(K_{r},K_{c}, P, \Pi_{r}, \Pi_{c})$, there exist an unique $K_{r}\times K_{r}$ matrix $B_{r}$ and an unique $K_{c}\times K_{r}$ matrix $B_{c}$ such that
\begin{itemize}
\item $U_{r}=\Pi_{r}B_{r}$ where $B_{r}=U_{r}(\mathcal{I}_{r},:)$. Meanwhile, $U_{r}(i_{r},:)=U_{r}(\bar{i}_{r},:)$ when $\Pi_{r}(i_{r},:)=\Pi_{r}(\bar{i}_{r},:)$ for $i_{r},\bar{i}_{r}\in [n_{r}]$.
  \item $U_{c}=\Pi_{c}B_{c}$. Meanwhile, $U_{c}(i_{c},:)=U_{c}(\bar{i}_{c},:)$ when $\ell(i_{c})=\ell(\bar{i}_{c})$ for $i_{c},\bar{i}_{c}\in[n_{c}]$, i.e., $U_{c}$ has $K_{c}$ distinct rows. Furthermore, when $K_{r}=K_{c}=K$, we have $\|B_{c}(k,:)-B_{c}(l,:)\|_{F}=\sqrt{\frac{1}{n_{c,k}}+\frac{1}{n_{c,l}}}$ for all $1\leq k<l\leq K$.
\end{itemize}
\end{lem}
Lemma \ref{RK} says that the rows of $U_{c}$ form a $K_{r}$-simplex in $\mathbb{R}^{K_{r}}$ which we call the Ideal Simplex (IS), with the $K_{r}$ rows of $B_{r}$ being the vertices. Such IS is also found in \cite{MixedSCORE,mao2020estimating,qing2021DiMMSB}. Meanwhile, Lemma \ref{RK} says that $U_{c}$ has $K_{c}$ distinct rows, and if two column nodes $i_{c}$ and $\bar{i}_{c}$ are from the same column community, then $U_{c}(i_{c},:)=U_{c}(\bar{i}_{c},:)$.

Under ONM, to recover $\Pi_{c}$ from $U_{c}$, since $U_{c}$ has $K_{c}$ distinct rows, applying k-means algorithm on all rows of $U_{c}$  returns true column communities by Lemma \ref{RK}. Meanwhile, since $U_{c}$ has $K_{c}$ distinct rows, we can set $\delta_{c}=\mathrm{min}_{k\neq l}\|B_{c}(k,:)-B_{c}(l,:)\|_{F}$ to measure the minimum center separation of $B_{c}$. By Lemma \ref{RK}, $\delta_{c}\geq \sqrt{\frac{2}{n_{c,\mathrm{max}}}}$ when $K_{r}=K_{c}=K$ under $ONM_{n_{r},n_{c}}(K_{r},K_{c}, P, \Pi_{r}, \Pi_{c})$. However, when $K_{r}<K_{c}$, it is challenge to obtain a positive lower bound of $\delta_{c}$, see the proof of Lemma \ref{RK} for detail.

Under ONM, to recover $\Pi_{c}$ from $U_{c}$, since $B_{r}$ is full rank, if $U_{r}$ and $B_{r}$ are known in advance ideally, we can exactly recover  $\Pi_{r}$ by setting $\Pi_{r}=U_{r}B_{r}'(B_{r}B_{r}')^{-1}$ by Lemma \ref{RK}. Set
$Y_{r}=U_{r}B_{r}'(B_{r}B_{r}')^{-1}$, since $Y_{r}\equiv\Pi_{r}$ and $\|\Pi_{r}(i_{r},:)\|_{1}=1$ for $i_{r\in[n_{r}]}$, we have
\begin{align*}
\Pi_{r}(i_{r},:)=\frac{Y_{r}(i_{r},:)}{\|Y_{r}(i_{r},:)\|_{1}}, i_{r}\in[n_{r}].
\end{align*}
With given $U_{r}$, since it enjoys IS structure $U_{r}=\Pi_{r}B_{r}\equiv \Pi_{r}U_{r}(\mathcal{I}_{r},:)$, as long as we can obtain the row corner matrix $U_{r}(\mathcal{I}_{r},:)$ (i.e., $B_{r}$), we can recover $\Pi_{r}$ exactly. As mentioned in \cite{MixedSCORE,mao2020estimating}, for such ideal simplex, the successive projection (SP) algorithm \cite{gillis2015semidefinite} (for detail of SP, see Algorithm \ref{alg:SP}) can be applied to $U_{r}$ with $K_{r}$ row communities to find $U_{r}(\mathcal{I}_{r},:)$.

Based on the above analysis, we are now ready to give the following algorithm which we call Ideal ONA. Input $\Omega, K_{r}, K_{c}$ with $K_{r}\leq K_{c}$. Output: $\Pi_{r}$ and $\ell$.
\begin{itemize}
  \item Let $\Omega=U_{r}\Lambda U'_{c}$ be the compact SVD of $\Omega$ such that $U_{r}\in\mathbb{R}^{n_{r}\times K_{r}},U_{c}\in\mathbb{R}^{n_{c}\times K_{r}}, \Lambda\in\mathbb{R}^{K_{r}\times K_{r}},U'_{r}U_{r}=I_{K_{r}},U'_{c}U_{c}=I_{K_{r}}$.
  \item For row nodes,
       \begin{itemize}
          \item Run SP algorithm on all rows of $U_{r}$ assuming there are $K_{r}$ row communities to obtain $U_{r}(\mathcal{I}_{r},:)$. Set $B_{r}=U_{r}(\mathcal{I}_{r},:)$.
          \item Set  $Y_{r}=U_{r}B_{r}'(B_{r}B_{r}')^{-1}$.  Recover $\Pi_{r}$ by setting $\Pi_{r}(i_{r},:)=\frac{Y_{r}(i_{r},:)}{\|Y_{r}(i_{r},:)\|_{1}}$ for $i_{r}\in[n_{r}]$.
        \end{itemize}
        For column nodes,
        \begin{itemize}
          \item Run k-means on $U_{c}$ assuming there are $K_{c}$ column communities, i.e., find the solution to the following optimization problem
              \begin{align*}
              M^{*}=\mathrm{argmin}_{M\in M_{n_{c}, K_{r}, K_{c}}}\|M-U_{c}\|^{2}_{F},
              \end{align*}
              where $M_{n_{c}, K_{r}, K_{c}}$ denotes the set of $n_{c}\times K_{r}$ matrices with only $K_{c}$ different rows.
          \item use $M^{*}$ to obtain the labels vector $\ell$ of column nodes.
        \end{itemize}
\end{itemize}
Follow similar proof of Theorem 1 of \cite{qing2021DiMMSB}, Ideal ONA exactly recoveries row nodes memberships and column nodes labels, and this also verifies the identifiability of ONM in turn. For convenience, call the two steps for column nodes as ``run k-means on $U_{c}$ assuming there are $K_{c}$ column communities to obtain $\ell$''.

We now extend the ideal case to the real case. Set $\tilde{A}=\hat{U}_{r}\hat{\Lambda}\hat{U}'_{c}$ be the top-$K_{r}$-dimensional SVD of $A$ such that $\hat{U}_{r}\in \mathbb{R}^{n_{r}\times K_{r}}, \hat{U}_{c}\in \mathbb{R}^{n_{c}\times K_{r}}, \hat{\Lambda}\in \mathbb{R}^{K_{r}\times K_{r}},\hat{U}'_{r}\hat{U}_{r}=I_{K_{r}}, \hat{U}'_{c}\hat{U}_{c}=I_{K_{r}}$, and $\hat{\Lambda}$ contains the top $K_{r}$ singular values of $A$. For the real case, we use $\hat{B}_{r}, \hat{B}_{c},\hat{Y}_{r}, \hat{\Pi}_{r}, \hat{\Pi}_{c}$ given in Algorithm \ref{alg:ONA}  to estimate $B_{r}, B_{c},Y_{r}, \Pi_{r},\Pi_{c}$, respectively.  Algorithm \ref{alg:ONA} called overlapping and nonoverlapping algorithm (ONA for short) is a natural extension of the Ideal ONA to the real case. In ONA, we set the negative entries of $\hat{Y}_{r}$ as 0 by setting $\hat{Y}_{r}=\mathrm{max}(0, \hat{Y}_{r})$ for the reason that weights for any row node should be nonnegative while there may exist some negative entries of $\hat{U}_{r}\hat{B}_{r}'(\hat{B}_{r}\hat{B}_{r}')^{-1}$. Note that, in a directed network, if column nodes have overlapping property while row nodes do not, to do community detection for such directed network, set the transpose of the adjacency matrix as input when applying our algorithm.
\begin{algorithm}
\caption{\textbf{Overlapping and Nonoverlapping Algorithm} (\textbf{ONA})}
\label{alg:ONA}
\begin{algorithmic}[1]
\Require The adjacency matrix $A\in \mathbb{R}^{n_{r}\times n_{c}}$ of a directed network, the number of row communities $K_{r}$, and the number of column communities $K_{c}$ with $K_{r}\leq K_{c}$.
\Ensure The estimated $n_{r}\times K_{r}$ membership matrix $\hat{\Pi}_{r}$ for row nodes, and the estimated $n_{c}\times 1$ labels vector $\hat{\ell}$ for column nodes.
\State Compute $\hat{U}_{r}\in\mathbb{R}^{n_{r}\times K_{r}}$ and $\hat{U}_{c}\in \mathbb{R}^{n_{c}\times K_{r}}$ from the top-$K_{r}$-dimensional SVD of $A$.
\State For row nodes:
\begin{itemize}
  \item Apply SP algorithm (i.e., Algorithm \ref{alg:SP}) on the rows of $\hat{U}_{r}$ assuming there are $K_{r}$ row clusters to obtain the near-corners matrix $\hat{U}_{r}(\mathcal{\hat{I}}_{r},:)\in\mathbb{R}^{K_{r}\times K_{r}}$, where $\mathcal{\hat{I}}_{r}$ is the index set returned by SP algorithm. Set $\hat{B}_{r}=\hat{U}_{r}(\mathcal{\hat{I}}_{r},:)$.
  \item Compute the $n_{r}\times K_{r}$ matrix $\hat{Y}_{r}$ such that $\hat{Y}_{r}=\hat{U}_{r}\hat{B}_{r}'(\hat{B}_{r}\hat{B}_{r}')^{-1}$. Set $\hat{Y}_{r}=\mathrm{max}(0, \hat{Y}_{r})$ and estimate $\Pi_{r}(i_{r},:)$ by $\hat{\Pi}_{r}(i_{r},:)=\frac{\hat{Y}_{r}(i_{r},:)}{\|\hat{Y}_{r}(i_{r},:)\|_{1}}, i_{r}\in[n_{r}]$.
\end{itemize}
For column nodes: run k-means on $\hat{U}_{c}$ assuming there are $K_{c}$ column communities to obtain $\hat{\ell}$.
\end{algorithmic}
\end{algorithm}
\subsection{Main results for ONA}
In this section, we show the consistency of our algorithm for fitting the ONM as the number of row nodes $n_{r}$ and the number of column nodes $n_{c}$ increase. Throughout this paper, $K_{r}\leq K_{c}$ are two known integers. First, we assume that
\begin{assum}\label{a1}
$\rho \mathrm{max}(n_{r},n_{c})\geq \mathrm{log}(n_{r}+n_{c})$.
\end{assum}
Assumption (\ref{a1}) controls the sparsity of directed network considered for theoretical study. By Lemma 4 of \cite{qing2021DiMMSB}, we have below lemma.
\begin{lem}\label{rowwiseerror}
	(Row-wise singular eigenvector error) Under $ONM_{n_{r},n_{c}}(K_{r},K_{c}, P, \Pi_{r}, \Pi_{c})$, when Assumption (\ref{a1}) holds, suppose $\sigma_{K_{r}}(\Omega)\geq C\sqrt{\rho (n_{r}+n_{c})\mathrm{log}(n_{r}+n_{c})}$, with probability at least $1-o((n_{r}+n_{c})^{-\alpha})$,
\begin{align*}
&\|\hat{U}_{r}\hat{U}'_{r}-U_{r}U'_{r}\|_{2\rightarrow\infty}=O(\frac{\sqrt{K_{r}}(\kappa(\Omega)\sqrt{\frac{\mathrm{max}(n_{r},n_{c})\mu}{\mathrm{min}(n_{r},n_{c})}}+\sqrt{\mathrm{log}(n_{r}+n_{c})})}{\sqrt{\rho}\sigma_{K_{r}}(\tilde{P})\sigma_{K_{r}}(\Pi_{r})\sqrt{n_{c,K_{r}}}}),
\end{align*}
where $\mu$ is the incoherence parameter defined as $\mu=\mathrm{max}(\frac{n_{r}\|U_{r}\|^{2}_{2\rightarrow\infty}}{K_{r}},\frac{n_{c}\|U_{c}\|^{2}_{2\rightarrow\infty}}{K_{r}})$.
\end{lem}
For convenience, set $\varpi=\|\hat{U}_{r}\hat{U}'_{r}-U_{r}U'_{r}\|_{2\rightarrow\infty}$ in this paper. To measure the performance of ONA for row nodes memberships, since row nodes have mixed memberships, naturally, we use the $l_{1}$ norm difference between $\Pi_{r}$ and $\hat{\Pi}_{r}$. Since column nodes are all pure nodes, we consider the performance criterion defined in \cite{joseph2016impact} to measure estimation error of ONA on column nodes. We introduce this measurement of estimation error as below.

Let $\mathcal{T}_{c}=\{\mathcal{T}_{c,1}, \mathcal{T}_{c,2}, \ldots, \mathcal{T}_{c,K_{c}}\}$ be the true partition of column nodes $\{1,2,\ldots, n_{c}\}$ obtained from $\ell$ such that $\mathcal{T}_{c,k}=\{i_{c}: \ell(i_{c})=k\}$ for $k\in[K_{c}]$. Let $\mathcal{\hat{T}}_{c}=\{\mathcal{\hat{T}}_{c,1}, \mathcal{\hat{T}}_{c,2}, \ldots, \mathcal{\hat{T}}_{c,K_{c}}\}$ be the estimated partition of column nodes $\{1,2,\ldots, n_{c}\}$ obtained from $\hat{\ell}$ of ONA such that $\mathcal{\hat{T}}_{c,k}=\{i_{c}: \hat{\ell}(i_{c})=k\}$ for $k\in[K_{c}]$. The criterion is defined as
\begin{align*}
\hat{f}_{c}=\mathrm{min}_{\pi\in S_{K_{c}}}\mathrm{max}_{k\in[K_{c}]}\frac{|\mathcal{T}_{c,k}\cap \mathcal{\hat{T}}^{c}_{c,\pi(k)}|+|\mathcal{T}^{c}_{c,k}\cap \mathcal{\hat{T}}_{c,\pi(k)}|}{n_{c,k}},
\end{align*}
where $S_{K_{c}}$ is the set of all permutations of $\{1,2,\ldots, K_{c}\}$ and the superscript $c$ denotes complementary set. As mentioned in \cite{joseph2016impact}, $\hat{f}_{c}$ measures the maximum proportion of column nodes in the symmetric difference of $\mathcal{T}_{c,k}$ and $\mathcal{\hat{T}}_{c,\pi(k)}$.

Next theorem gives theoretical bounds on estimations of memberships for both row and column nodes, which is the main theoretical result for ONA.
\begin{thm}\label{Main}
Under $ONM_{n_{r},n_{c}}(K_{r},K_{c}, P, \Pi_{r}, \Pi_{c})$, suppose conditions in Lemma \ref{rowwiseerror} hold, with probability at least $1-o((n_{r}+n_{c})^{-\alpha})$,
\begin{itemize}
  \item for row nodes, there exists a permutation matrix $\mathcal{P}_{r}$ such that
  \begin{align*}
\mathrm{max}_{i_{r}\in[n_{r}]}\|e'_{i_{r}}(\hat{\Pi}_{r}-\Pi_{r}\mathcal{P}_{r})\|_{1}=O(\varpi\kappa(\Pi'_{r}\Pi_{r})K_{r}\sqrt{\lambda_{1}(\Pi'_{r}\Pi_{r})}).
  \end{align*}
  \item for column nodes,
  \begin{align*}
    \hat{f}_{c}=O(\frac{K_{r}K_{c}\mathrm{max}(n_{r},n_{c})\mathrm{log}(n_{r}+n_{c})}{\sigma^{2}_{K_{r}}(\tilde{P})\rho \delta^{2}_{c}\sigma^{2}_{K_{r}}(\Pi_{r})n_{c,K_{r}}n_{c,\mathrm{min}}}).
  \end{align*}
  Especially, when $K_{r}=K_{c}=K$,
\begin{align*}
\hat{f}_{c}=O(\frac{K^{2}\mathrm{max}(n_{r},n_{c})n_{c,\mathrm{max}}\mathrm{log}(n_{r}+n_{c})}{\sigma^{2}_{K}(\tilde{P})\rho\sigma^{2}_{K}(\Pi_{r})n^{2}_{c,\mathrm{min}}}).
\end{align*}
\end{itemize}
\end{thm}
Add conditions similar as Corollary 3.1 in \cite{mao2020estimating}, we have the following corollary.
\begin{cor}\label{AddConditions}
Under $ONM_{n_{r},n_{c}}(K_{r},K_{c}, P, \Pi_{r}, \Pi_{c})$, suppose conditions in Lemma \ref{rowwiseerror} hold, and further suppose that $\lambda_{K_{r}}(\Pi'_{r}\Pi_{r})=O(\frac{n_{r}}{K_{r}}), n_{c,\mathrm{min}}=O(\frac{n_{c}}{K_{c}})$, with probability at least $1-o((n_{r}+n_{c})^{-\alpha})$,
\begin{itemize}
  \item for row nodes, when $K_{r}=K_{c}=K$,
 \begin{align*}
\mathrm{max}_{i_{r}\in[n_{r}]}\|e'_{i_{r}}(\hat{\Pi}_{r}-\Pi_{r}\mathcal{P}_{r})\|_{1}=O(\frac{K^{2}(\sqrt{\frac{C\mathrm{max}(n_{r},n_{c})}{\mathrm{min}(n_{r},n_{c})}}+\sqrt{\mathrm{log}(n_{r}+n_{c})})}{\sigma_{K}(\tilde{P})\sqrt{\rho
n_{c}}}).
\end{align*}
  \item for column nodes,
  \begin{align*}
    \hat{f}_{c}=O(\frac{K^{2}_{r}K^{3}_{c}\mathrm{max}(n_{r},n_{c})\mathrm{log}(n_{r}+n_{c})}{\sigma^{2}_{K_{r}}(\tilde{P})\rho \delta^{2}_{c}n_{r}n^{2}_{c}}).
  \end{align*}
  When $K_{r}=K_{c}=K$,
\begin{align*}
\hat{f}_{c}=O(\frac{K^{4}\mathrm{max}(n_{r},n_{c})\mathrm{log}(n_{r}+n_{c})}{\sigma^{2}_{K}(\tilde{P})\rho n_{r}n_{c}}).
\end{align*}
\end{itemize}
Especially, when $n_{r}=O(n), n_{c}=O(n), K_{r}=O(1)$ and $K_{c}=O(1)$,
\begin{itemize}
  \item for row nodes, when $K_{r}=K_{c}$,
 \begin{align*}
\mathrm{max}_{i_{r}\in[n_{r}]}\|e'_{i_{r}}(\hat{\Pi}_{r}-\Pi_{r}\mathcal{P}_{r})\|_{1}=O(\frac{\sqrt{\mathrm{log}(n)}}{\sigma_{K_{r}}(\tilde{P})\sqrt{\rho
n}}).
\end{align*}
  \item for column nodes,
  \begin{align*}
    \hat{f}_{c}=O(\frac{\mathrm{log}(n)}{\sigma^{2}_{K_{r}}(\tilde{P})\rho \delta^{2}_{c}n^{2}}).
  \end{align*}
  When $K_{r}=K_{c}=K$,
\begin{align*}
\hat{f}_{c}=O(\frac{\mathrm{log}(n)}{\sigma^{2}_{K}(\tilde{P})\rho n}).
\end{align*}
\end{itemize}
\end{cor}
When $K_{r}\neq K_{c}$, though it is challenge to obtain the lower bound of $\delta_{c}$, we can roughly set $\sqrt{\frac{2}{n_{c,\mathrm{max}}}}$ as the lower bound of $\delta_{c}$ since $\delta_{c}\geq \sqrt{\frac{2}{n_{c,\mathrm{max}}}}$ when $K_{r}=K_{c}$.

When ONM degenerates to SBM by setting $\Pi_{r}=\Pi_{c}$ and all nodes are pure, applying the separation condition and sharp threshold criterion developed in \cite{qingSCSTC} on the upper bounds of error rates in Corollary \ref{AddConditions}, sure we can obtain the classical separation condition of a balanced network and sharp threshold of the Erd\"os-R\'enyi random graph $G(n,p)$ of \cite{erdos2011on}, and this guarantees the optimality of our theoretical results.
\section{The overlapping and degree-corrected nonoverlapping model}
Similar as DCSBM \cite{DCSBM} is an extension of SBM by introducing node  specific parameters to allow for varying degrees, in this section, we propose an extension of ONM by considering degree heterogeneity and build theoretical guarantees for algorithm fitting our model.

Let $\theta_{c}$ be an $n_{c}\times 1$ vector whose $i_{c}$-th entry is the degree heterogeneity of column node $i_{c}$, for $i_{c}\in[n_{c}]$. Let $\Theta_{c}$ be an $n_{c}\times n_{c}$ diagonal matrix whose $i_{c}$-th diagonal element is $\theta_{c}(i_{c})$. The extended model for generating $A$ is as follows:
\begin{align}\label{ONMODCNM}
\Omega:=\Pi_{r}P\Pi'_{c}\Theta_{c},~~~A(i_{r},i_{c})\sim\mathrm{Bernoulli}(\Omega(i_{r},i_{c}))\qquad \mathrm{for~}i_{r}\in[n_{r}],i_{c}\in[n_{c}].
\end{align}

\begin{defin}
Call model (\ref{Pir}), (\ref{Pic}), (\ref{krkc}),(\ref{definP}), (\ref{ONMODCNM}) the Overlapping and Degree-Corrected Nonoverlapping model (ODCNM) and denote it by $ODCNM_{n_{r},n_{c}}(K_{r},K_{c}, P, \Pi_{r}, \Pi_{c},\Theta_{c})$. 
\end{defin}
Note that, under ODCNM, the maximum element of $P$ can be larger than 1 since $\mathrm{max}_{i_{c}\in[n_{c}]}\theta_{c}(i_{c})$ also can control the sparsity of the directed network $\mathcal{N}$. The following proposition guarantees that ODCNM is identifiable in terms of $P, \Pi_{r}$ and $\Pi_{c}$, and such identifiability is similar as that of DCSBM and DCScBM. 
\begin{prop}\label{idODCNM}
	If conditions (I1) and (I2) hold, ODCNM is identifiable for membership matrices: For eligible $(P,\Pi_{r}, \Pi_{c},\Theta_{c})$ and $(\check{P},\check{\Pi}_{r}, \check{\Pi}_{c},\check{\Theta}_{c})$, if $\Pi_{r}P\Pi'_{c}\Theta_{c}=\check{\Pi}_{r}\check{P}\check{\Pi}'_{c}\check{\Theta}_{c}$, then $\Pi_{r}=\check{\Pi}_{r}$ and $\Pi_{c}=\check{\Pi}_{c}$.
\end{prop}
\begin{rem}
By setting $\theta_{c}(i_{c})=\rho$ for $i_{c}\in[n_{c}]$, ODCNM reduces to ONM, and this is the reason that ODCNM can be seen as an extension of ONM. Meanwhile, though DCScBM \cite{DISIM} can model directed networks with degree heterogeneities for both row and column nodes, DCScBM does not allow the overlapping property for nodes. For comparison, our ODCNM allows row nodes have overlapping property at the cost of losing the degree heterogeneities and requiring $K_{r}\leq K_{c}$ for model identifiability. Furthermore, another identifiable model extends ONM by considering degree heterogeneity for row nodes with overlapping property is provided in Appendix \ref{DCOOOOOONM}, in which we also explain why we do not extend ONM by considering degree heterogeneities for both row and column nodes.
\end{rem}
\subsection{A spectral algorithm for fitting ODCNM}
We now discuss our intuition for the design of our algorithm to fit ODCNM. Without causing confusion, we also use $U_{r}, U_{c}, B_{r}, B_{c}, \delta_{c}, Y_{r}$, and so on under ODCNM.  Let $U_{c,*}\in \mathbb{R}^{n_{c}\times K_{r}}$ be the row-normalized version of $U_{c}$ such that $U_{c,*}(i_{c},:)=\frac{U_{c}(i_{c},:)}{\|U_{c}(i_{c},:)\|_{F}}$ for $i_{c}\in[n_{c}]$. Then clustering the rows of $U_{c,*}$ by k-means algorithm can return perfect clustering for column nodes, and this is guaranteed by next lemma.
\begin{lem}\label{RKODCNM}
Under $ODCNM_{n_{r},n_{c}}(K_{r},K_{c}, P, \Pi_{r}, \Pi_{c},\Theta_{c})$, there exist an unique $K_{r}\times K_{r}$ matrix $B_{r}$ and an unique $K_{c}\times K_{r}$ matrix $B_{c}$ such that
\begin{itemize}
\item $U_{r}=\Pi_{r}B_{r}$ where $B_{r}=U_{r}(\mathcal{I}_{r},:)$. Meanwhile, $U_{r}(i_{r},:)=U_{r}(\bar{i}_{r},:)$ when $\Pi_{r}(i_{r},:)=\Pi_{r}(\bar{i}_{r},:)$ for $i_{r},\bar{i}_{r}\in [n_{r}]$.
  \item $U_{c,*}=\Pi_{c}B_{c}$. Meanwhile, $U_{c,*}(i_{c},:)=U_{c,*}(\bar{i}_{c},:)$ when $\ell(i_{c})=\ell(\bar{i}_{c})$ for $i_{c},\bar{i}_{c}\in[n_{c}]$. Furthermore, when $K_{r}=K_{c}=K$, we have $\|B_{c}(k,:)-B_{c}(l,:)\|_{F}=\sqrt{2}$ for all $1\leq k<l\leq K$.
\end{itemize}
\end{lem}
Recall that we set $\delta_{c}=\mathrm{min}_{k\neq l}\|B_{c}(k,:)-B_{c}(l,:)\|_{F}$, by Lemma \ref{RKODCNM}, $\delta_{c}=\sqrt{2}$ when $K_{r}=K_{c}=K$ under $ODCNM_{n_{r},n_{c}}(K_{r},K_{c}, P, \Pi_{r}, \Pi_{c},\Theta_{c})$. However, when $K_{r}<K_{c}$, it is challenge to obtain a positive lower bound of $\delta_{c}$, see the proof of Lemma \ref{RKODCNM} for detail.

Under ODCNM, to recover $\Pi_{c}$ from $U_{c}$, since $U_{c,*}$ has $K_{c}$ distinct rows, applying k-means algorithm on all rows of $U_{c,*}$  returns true column communities by Lemma \ref{RKODCNM}; to recover $\Pi_{r}$ from $U_{r}$,  just follow same idea as that of under ONM.

Based on the above analysis, we are now ready to give the following algorithm which we call Ideal ODCNA. Input $\Omega, K_{r}, K_{c}$ with $K_{r}\leq K_{c}$. Output: $\Pi_{r}$ and $\ell$.
\begin{itemize}
  \item Let $\Omega=U_{r}\Lambda U'_{c}$ be the compact SVD of $\Omega$ such that $U_{r}\in\mathbb{R}^{n_{r}\times K_{r}},U_{c}\in\mathbb{R}^{n_{c}\times K_{r}}, \Lambda\in\mathbb{R}^{K_{r}\times K_{r}},U'_{r}U_{r}=I_{K_{r}},U'_{c}U_{c}=I_{K_{r}}$. Let $U_{c,*}$ be the row-normalization of $U_{c}$.
  \item For row nodes,
       \begin{itemize}
          \item Run SP algorithm on all rows of $U_{r}$ assuming there are $K_{r}$ row communities to obtain $U_{r}(\mathcal{I}_{r},:)$. Set $B_{r}=U_{r}(\mathcal{I}_{r},:)$.
          \item Set  $Y_{r}=U_{r}B_{r}'(B_{r}B_{r}')^{-1}$.  Recover $\Pi_{r}$ by setting $\Pi_{r}(i_{r},:)=\frac{Y_{r}(i_{r},:)}{\|Y_{r}(i_{r},:)\|_{1}}$ for $i_{r}\in[n_{r}]$.
        \end{itemize}
        For column nodes: run k-means on $U_{c,*}$ assuming there are $K_{c}$ column communities to obtain $\ell$.
\end{itemize}
Sure, Ideal ODCNA exactly recoveries row nodes memberships and column nodes labels, and this also supports the identifiability of ODCNM.

We now extend the ideal case to the real case. Let $\hat{U}_{c,*}\in\mathbb{R}^{n_{c}\times K_{r}}$ be the row-normalized version of $\hat{U}_{c}$ such that $\hat{U}_{c,*}(i_{c},:)=\frac{\hat{U}_{c}(i_{c},:)}{\|\hat{U}_{c}(i_{c},:)\|_{F}}$ for $i_{c}\in[n_{c}]$. Algorithm \ref{alg:ODCNA} called overlapping and degree-corrected nonoverlapping algorithm (ODCNA for short) is a natural extension of the Ideal ODCNA to the real case.
\begin{algorithm}
\caption{\textbf{Overlapping and Degree-Corrected Nonoverlapping Algorithm} (\textbf{ODCNA})}
\label{alg:ODCNA}
\begin{algorithmic}[1]
\Require The adjacency matrix $A\in \mathbb{R}^{n_{r}\times n_{c}}$ of a directed network, the number of row communities $K_{r}$, and the number of column communities $K_{c}$ with $K_{r}\leq K_{c}$.
\Ensure The estimated $n_{r}\times K_{r}$ membership matrix $\hat{\Pi}_{r}$ for row nodes, and the estimated $n_{c}\times 1$ labels vector $\hat{\ell}$ for column nodes.
\State Compute $\hat{U}_{r}\in\mathbb{R}^{n_{r}\times K_{r}}$ and $\hat{U}_{c}\in \mathbb{R}^{n_{c}\times K_{r}}$ from the top-$K_{r}$-dimensional SVD of $A$. Compute $\hat{U}_{c,*}$ from $\hat{U}_{c}$.
\State For row nodes:
\begin{itemize}
  \item Apply SP algorithm (i.e., Algorithm \ref{alg:SP}) on the rows of $\hat{U}_{r}$ assuming there are $K_{r}$ row clusters to obtain the near-corners matrix $\hat{U}_{r}(\mathcal{\hat{I}}_{r},:)\in\mathbb{R}^{K_{r}\times K_{r}}$, where $\mathcal{\hat{I}}_{r}$ is the index set returned by SP algorithm. Set $\hat{B}_{r}=\hat{U}_{r}(\mathcal{\hat{I}}_{r},:)$.
  \item Compute the $n_{r}\times K_{r}$ matrix $\hat{Y}_{r}$ such that $\hat{Y}_{r}=\hat{U}_{r}\hat{B}_{r}'(\hat{B}_{r}\hat{B}_{r}')^{-1}$. Set $\hat{Y}_{r}=\mathrm{max}(0, \hat{Y}_{r})$ and estimate $\Pi_{r}(i_{r},:)$ by $\hat{\Pi}_{r}(i_{r},:)=\frac{\hat{Y}_{r}(i_{r},:)}{\|\hat{Y}_{r}(i_{r},:)\|_{1}}, i_{r}\in[n_{r}]$.
\end{itemize}
For column nodes: run k-means on $\hat{U}_{c,*}$ assuming there are $K_{c}$ column communities to obtain $\hat{\ell}$.
\end{algorithmic}
\end{algorithm}
\subsection{Main results for ODCNA}
Set $\theta_{c,\mathrm{max}}=\mathrm{max}_{i_{c}\in[n_{c}]}\theta_{c}(i_{c}), \theta_{c,\mathrm{min}}=\mathrm{min}_{i_{c}\in[n_{c}]}\theta_{c}(i_{c})$, and $P_{\mathrm{max}}=\mathrm{max}_{k\in[K_{r}], l\in[n_{c}]}P(k,l)$. Assume that
\begin{assum}\label{a2}
$P_{\mathrm{max}}\mathrm{max}(\theta_{c,\mathrm{max}}n_{r},\|\theta_{c}\|_{1})\geq \mathrm{log}(n_{r}+n_{c})$.
\end{assum}
By the proof of Lemma 4.3 of \cite{qing2021DiDCMM}, we have below lemma.
\begin{lem}\label{rowwiseerrorODCNM}
	(Row-wise singular eigenvector error) Under $ODCNM_{n_{r},n_{c}}(K_{r},K_{c}, P, \Pi_{r}, \Pi_{c},\Theta_{c})$, when Assumption (\ref{a2}) holds, suppose $\sigma_{K_{r}}(\Omega)\geq C\sqrt{\theta_{c,\mathrm{max}}(n_{r}+n_{c})\mathrm{log}(n_{r}+n_{c})}$, with probability at least $1-o((n_{r}+n_{c})^{-\alpha})$,
\begin{align*}
&\|\hat{U}_{r}\hat{U}'_{r}-U_{r}U'_{r}\|_{2\rightarrow\infty}=O(\frac{\sqrt{\theta_{c,\mathrm{max}}K_{r}}(\kappa(\Omega)\sqrt{\frac{\mathrm{max}(n_{r},n_{c})\mu}{\mathrm{min}(n_{r},n_{c})}}+\sqrt{\mathrm{log}(n_{r}+n_{c})})}{\theta_{c,\mathrm{min}}\sigma_{K_{r}}(P)\sigma_{K_{r}}(\Pi_{r})\sqrt{n_{c,K_{r}}}}).
\end{align*}
\end{lem}
Next theorem is the main theoretical result for ODCNA, where we also use same measurements as ONA to measure the performances of ODCNA.
\begin{thm}\label{MainODCNM}
Under $ODCNM_{n_{r},n_{c}}(K_{r},K_{c}, P, \Pi_{r}, \Pi_{c},\Theta_{c})$, suppose conditions in Lemma \ref{rowwiseerrorODCNM} hold, with probability at least $1-o((n_{r}+n_{c})^{-\alpha})$,
\begin{itemize}
  \item for row nodes,
  \begin{align*}
\mathrm{max}_{i_{r}\in[n_{r}]}\|e'_{i_{r}}(\hat{\Pi}_{r}-\Pi_{r}\mathcal{P}_{r})\|_{1}=O(\varpi\kappa(\Pi'_{r}\Pi_{r})K_{r}\sqrt{\lambda_{1}(\Pi'_{r}\Pi_{r})}).
  \end{align*}
  \item for column nodes,
  \begin{align*}
\hat{f}_{c}=O(\frac{\theta^{2}_{c,\mathrm{max}}K_{r}K_{c}\mathrm{max}(
\theta_{c,\mathrm{max}}n_{r},\|\theta_{c}\|_{1})n_{c,\mathrm{max}}\mathrm{log}(n_{r}+n_{c})}{\sigma^{2}_{K_{r}}(P)\theta^{4}_{c,\mathrm{min}}\delta^{2}_{c}m^{2}_{V_{c}}\sigma^{2}_{K_{r}}(\Pi_{r})n_{c,K_{r}}n_{c,\mathrm{min}}}),
  \end{align*}
where $m_{V_{c}}$ is a parameter defined in the proof of this theorem, and it is $1$ when $K_{r}=K_{c}$. Especially, when $K_{r}=K_{c}=K$,
\begin{align*}
\hat{f}_{c}=O(\frac{\theta^{2}_{c,\mathrm{max}}K^{2}\mathrm{max}(
\theta_{c,\mathrm{max}}n_{r},\|\theta_{c}\|_{1})n_{c,\mathrm{max}}\mathrm{log}(n_{r}+n_{c})}{\sigma^{2}_{K}(P)\theta^{4}_{c,\mathrm{min}}\sigma^{2}_{K}(\Pi_{r})n^{2}_{c,\mathrm{min}}}).
\end{align*}
\end{itemize}
\end{thm}
Add some conditions on model parameters, we have the following corollary.
\begin{cor}\label{AddConditionsODCNM}
Under $ODCNM_{n_{r},n_{c}}(K_{r},K_{c}, P, \Pi_{r}, \Pi_{c},\Theta_{c})$, suppose conditions in Lemma \ref{rowwiseerrorODCNM} hold, and further suppose that $\lambda_{K_{r}}(\Pi'_{r}\Pi_{r})=O(\frac{n_{r}}{K_{r}}), n_{c,\mathrm{min}}=O(\frac{n_{c}}{K_{c}})$, with probability at least $1-o((n_{r}+n_{c})^{-\alpha})$,
\begin{itemize}
  \item for row nodes, when $K_{r}=K_{c}=K$,
\begin{align*}
\mathrm{max}_{i_{r}\in[n_{r}]}\|e'_{i_{r}}(\hat{\Pi}_{r}-\Pi_{r}\mathcal{P}_{r})\|_{1}=O(\frac{K^{2}\sqrt{\theta_{c,\mathrm{max}}}(\sqrt{\frac{C\mathrm{max}(n_{r},n_{c})}{\mathrm{min}(n_{r},n_{c})}}+\sqrt{\mathrm{log}(n_{r}+n_{c})})}{\theta_{c,\mathrm{min}}\sigma_{K}(P)\sqrt{n_{c}}}).
\end{align*}
  \item for column nodes,
  \begin{align*}
\hat{f}_{c}=O(\frac{\theta^{2}_{c,\mathrm{max}}K^{2}_{r}K^{2}_{c}\mathrm{max}(
\theta_{c,\mathrm{max}}n_{r},\|\theta_{c}\|_{1})\mathrm{log}(n_{r}+n_{c})}{\sigma^{2}_{K_{r}}(P)\theta^{4}_{c,\mathrm{min}}\delta^{2}_{c}m^{2}_{V_{c}}n_{r}n_{c}}).
  \end{align*}
  When $K_{r}=K_{c}=K$,
\begin{align*}
\hat{f}_{c}=O(\frac{\theta^{2}_{c,\mathrm{max}}K^{4}\mathrm{max}(
\theta_{c,\mathrm{max}}n_{r},\|\theta_{c}\|_{1})\mathrm{log}(n_{r}+n_{c})}{\sigma^{2}_{K}(P)\theta^{4}_{c,\mathrm{min}}n_{r}n_{c}}).
\end{align*}
\end{itemize}
Especially, when $n_{r}=O(n), n_{c}=O(n), K_{r}=O(1)$ and $K_{c}=O(1)$,
\begin{itemize}
  \item for row nodes, when $K_{r}=K_{c}$,
\begin{align*}
\mathrm{max}_{i_{r}\in[n_{r}]}\|e'_{i_{r}}(\hat{\Pi}_{r}-\Pi_{r}\mathcal{P}_{r})\|_{1}=O(\frac{\sqrt{\theta_{c,\mathrm{max}}\mathrm{log}(n)}}{\theta_{c,\mathrm{min}}\sigma_{K}(P)\sqrt{n}}).
\end{align*}
  \item for column nodes,
  \begin{align*}
\hat{f}_{c}=O(\frac{\theta^{2}_{c,\mathrm{max}}\mathrm{max}(
\theta_{c,\mathrm{max}}n_{r},\|\theta_{c}\|_{1})\mathrm{log}(n)}{\sigma^{2}_{K_{r}}(P)\theta^{4}_{c,\mathrm{min}}\delta^{2}_{c}m^{2}_{V_{c}}n^{2}}).
  \end{align*}
  When $K_{r}=K_{c}=K$,
\begin{align*}
\hat{f}_{c}=O(\frac{\theta^{2}_{c,\mathrm{max}}\mathrm{max}(
\theta_{c,\mathrm{max}}n_{r},\|\theta_{c}\|_{1})\mathrm{log}(n)}{\sigma^{2}_{K}(P)\theta^{4}_{c,\mathrm{min}}n^{2}}).
\end{align*}
\end{itemize}
\end{cor}
When $K_{r}\neq K_{c}$, though it is challenge to obtain the lower bounds of $\delta_{c}$ and $m_{V_{c}}$, we can roughly set $\sqrt{2}$ and $1$ as the lower bounds of $\delta_{c}$ and  $m_{V_{c}}$, respectively, since $\delta_{c}=\sqrt{2}$ and $m_{V_{c}}=1$ when $K_{r}=K_{c}$. Meanwhile, if we further set $\theta_{c,\mathrm{max}}=O(\rho)$ and $\theta_{c,\mathrm{min}}=O(\rho)$, we have below corollary.
\begin{cor}\label{AddRho}
Under $ODCNM_{n_{r},n_{c}}(K_{r},K_{c}, P, \Pi_{r}, \Pi_{c},\Theta_{c})$, suppose conditions in Lemma \ref{rowwiseerrorODCNM} hold, and further suppose that $\lambda_{K_{r}}(\Pi'_{r}\Pi_{r})=O(\frac{n_{r}}{K_{r}}), n_{c,\mathrm{min}}=O(\frac{n_{c}}{K_{c}})$ and $\theta_{c,\mathrm{max}}=O(\rho), \theta_{c,\mathrm{min}}=O(\rho)$, with probability at least $1-o((n_{r}+n_{c})^{-\alpha})$,
\begin{itemize}
  \item for row nodes, when $K_{r}=K_{c}=K$,
\begin{align*}
\mathrm{max}_{i_{r}\in[n_{r}]}\|e'_{i_{r}}(\hat{\Pi}_{r}-\Pi_{r}\mathcal{P}_{r})\|_{1}=O(\frac{K^{2}(\sqrt{\frac{C\mathrm{max}(n_{r},n_{c})}{\mathrm{min}(n_{r},n_{c})}}+\sqrt{\mathrm{log}(n_{r}+n_{c})})}{\sigma_{K}(P)\sqrt{\rho n_{c}}}).
\end{align*}
  \item for column nodes,
  \begin{align*}
\hat{f}_{c}=O(\frac{K^{2}_{r}K^{2}_{c}\mathrm{max}(n_{r},n_{c})\mathrm{log}(n_{r}+n_{c})}{\sigma^{2}_{K_{r}}(P)\rho\delta^{2}_{c}m^{2}_{V_{c}}n_{r}n_{c}}).
  \end{align*}
  When $K_{r}=K_{c}=K$,
\begin{align*}
\hat{f}_{c}=O(\frac{K^{4}\mathrm{max}(n_{r},n_{c})\mathrm{log}(n_{r}+n_{c})}{\sigma^{2}_{K}(P)\rho n_{r}n_{c}}).
\end{align*}
\end{itemize}
Especially, when $n_{r}=O(n), n_{c}=O(n), K_{r}=O(1)$ and $K_{c}=O(1)$,
\begin{itemize}
  \item for row nodes, when $K_{r}=K_{c}$,\
\begin{align*}
\mathrm{max}_{i_{r}\in[n_{r}]}\|e'_{i_{r}}(\hat{\Pi}_{r}-\Pi_{r}\mathcal{P}_{r})\|_{1}=O(\frac{\sqrt{\mathrm{log}(n)}}{\sigma_{K}(P)\sqrt{\rho n}}).
\end{align*}
  \item for column nodes,
  \begin{align*}
\hat{f}_{c}=O(\frac{\mathrm{log}(n)}{\sigma^{2}_{K_{r}}(P)\rho\delta^{2}_{c}m^{2}_{V_{c}}n}).
  \end{align*}
  When $K_{r}=K_{c}=K$,
\begin{align*}
\hat{f}_{c}=O(\frac{\mathrm{log}(n)}{\sigma^{2}_{K}(P)\rho n}).
\end{align*}
\end{itemize}
\end{cor}
By setting $\Theta_{c}=\rho I$, $ODCNM_{n_{r},n_{c}}(K_{r}, K_{c}, P,\Pi_{r}, \Pi_{c},\Theta_{c})$ degenerates to $ONM_{n_{r},n_{c}}(K_{r}, K_{c}, P,\Pi_{r}, \Pi_{c})$. By comparing Corollary \ref{AddConditions} and Corollary \ref{AddRho}, we see that theoretical results under ODCNM are consistent with those under ONM when ODCNM degenerates to ONM for the case that $K_{r}=K_{c}=K$.
\section{Simulations}\label{sec5}
In this section,we present some simulations to investigate the performance of the three proposed algorithms. We measure their performances by Mixed-Hamming error rate (MHamm for short) for row nodes and Hamming error rate (Hamm for short) for column nodes defined below
\begin{align*}
\mathrm{MHamm}=\frac{\mathrm{min}_{\pi\in S_{K_{r}}}\|\hat{\Pi}_{r}\pi-\Pi_{r}\|_{1}}{n_{r}}\mathrm{~and~} \mathrm{Hamm}=\frac{\mathrm{min}_{\pi\in S}\|\hat{\Pi}_{c}\pi-\Pi_{c}\|_{1}}{n_{c}},
\end{align*}
where $\hat{\Pi}_{c}\in\mathbb{R}^{n_{c}\times K_{c}}$ is defined as $\hat{\Pi}_{c}(i_{c},k)=1$ if  $\hat{\ell}(i_{c})=k$ and 0 otherwise for $i_{c}\in[n_{c}], k\in[K_{c}]$.

For all simulations in this section, the parameters $(n_{r}, n_{c}, K_{r}, K_{c}, P, \rho, \Pi_{r}, \Pi_{c}, \Theta_{c})$ are set as follows. Unless specified, set $n_{r}=400, n_{c}=300, K_{r}=3, K_{c}=4$. For column nodes, generate $\Pi_{c}$ by setting each column node belonging to one of the column communities with equal probability. Let each row community have $100$ pure nodes, and let all the mixed row nodes have memberships $(0.6,0.3,0.1)$. $P=\rho\tilde{P}$ is set independently under ONM and ODCNM. Under ONM, $\rho$ is 0.5 in Experiment 1 and we study the influence of $\rho$ in Experiment 2; Under ODCNM, for $z_{c}\geq 1$, we generate the degree parameters for column nodes as below: let $\theta_{c}\in\mathbb{R}^{n_{c}\times 1}$ such that $1/\theta_{c}(i_{c})\overset{iid}{\sim}U(1,z_{c})$ for $i_{c}\in[n_{c}]$, where $U(1,z_{c})$ denotes the uniform distribution on $[1, z_{c}]$. We study the influences of $Z_{c}$ and $\rho$ under ODCNM in Experiments 3 and 4, respectively.
For all settings, we report the averaged MHamm and the averaged Hamm over 50 repetitions.

\texttt{Experiment 1: Changing $n_{c}$ under ONM.} Let $n_{c}$ range in $\{50, 100, 150,\ldots, 300\}$. For this experiment, $P$ is set as
\[P=\rho\begin{bmatrix}
    1&0.3&0.2&0.3\\
    0.2&0.9&0.1&0.2\\
    0.3&0.2&0.8&0.3\\
\end{bmatrix}.\] Let $\rho=0.5$ under for this experiment designed under ONM. The numerical results are shown in panels (a) and (b) of Figure \ref{EX}. The results show that as $n_{c}$ increases, ONA and ODCNA perform better. Meanwhile, the total run-time for this experiment is roughly 70 seconds. For row nodes, since both ONA and ODCNA apply SP algorithm on $\hat{U}$ to estimate $\Pi_{r}$, the estimated row membership matrices of ONA and ODCNA are same, and hence MHamm for ONA always equal to that of ODCNA.
\begin{figure}
\centering
\subfigure[Changing $n_{c}$ under ONM: MHamm.]{\includegraphics[width=0.37\textwidth]{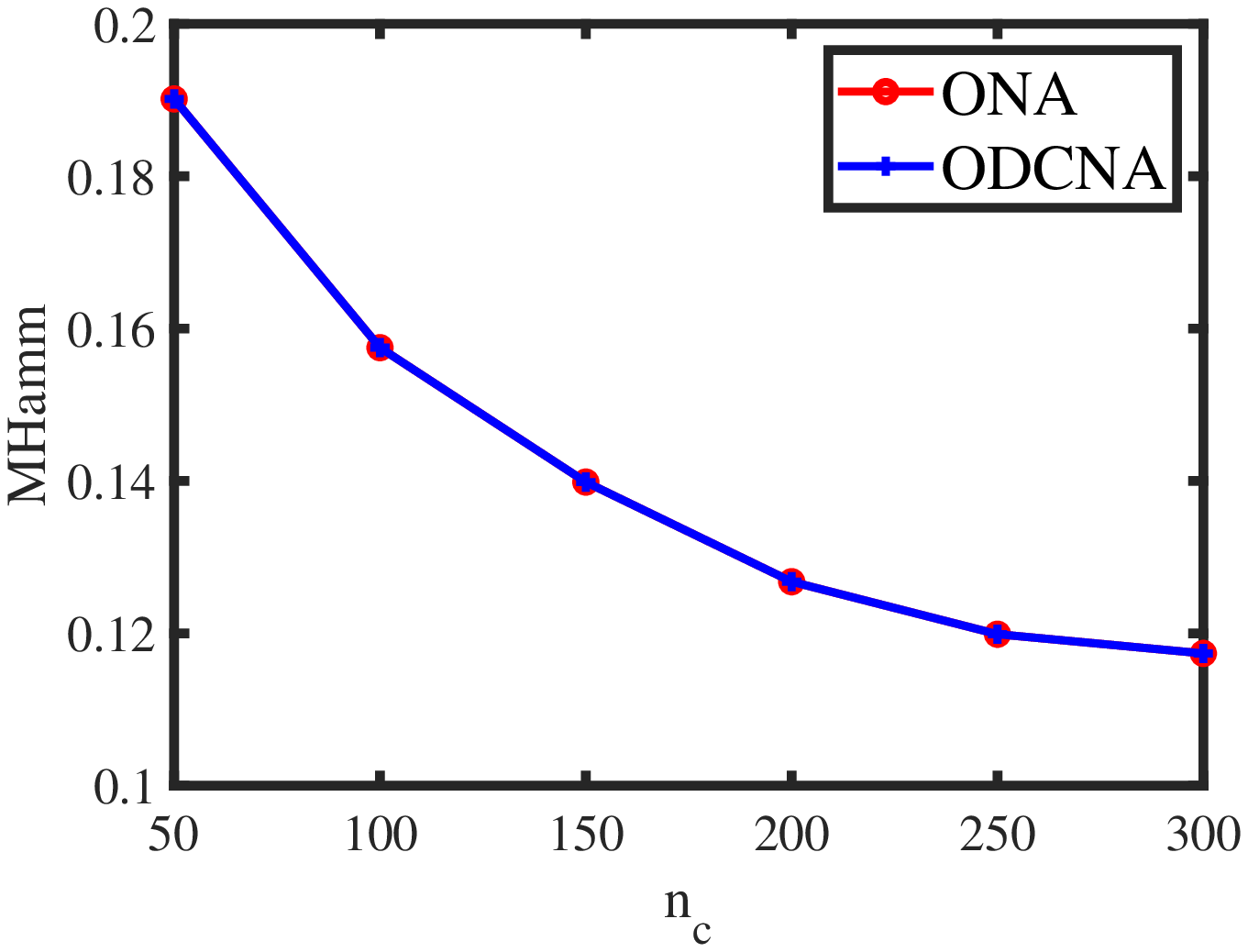}}
\subfigure[Changing $n_{c}$ under ONM: Hamm.]{\includegraphics[width=0.37\textwidth]{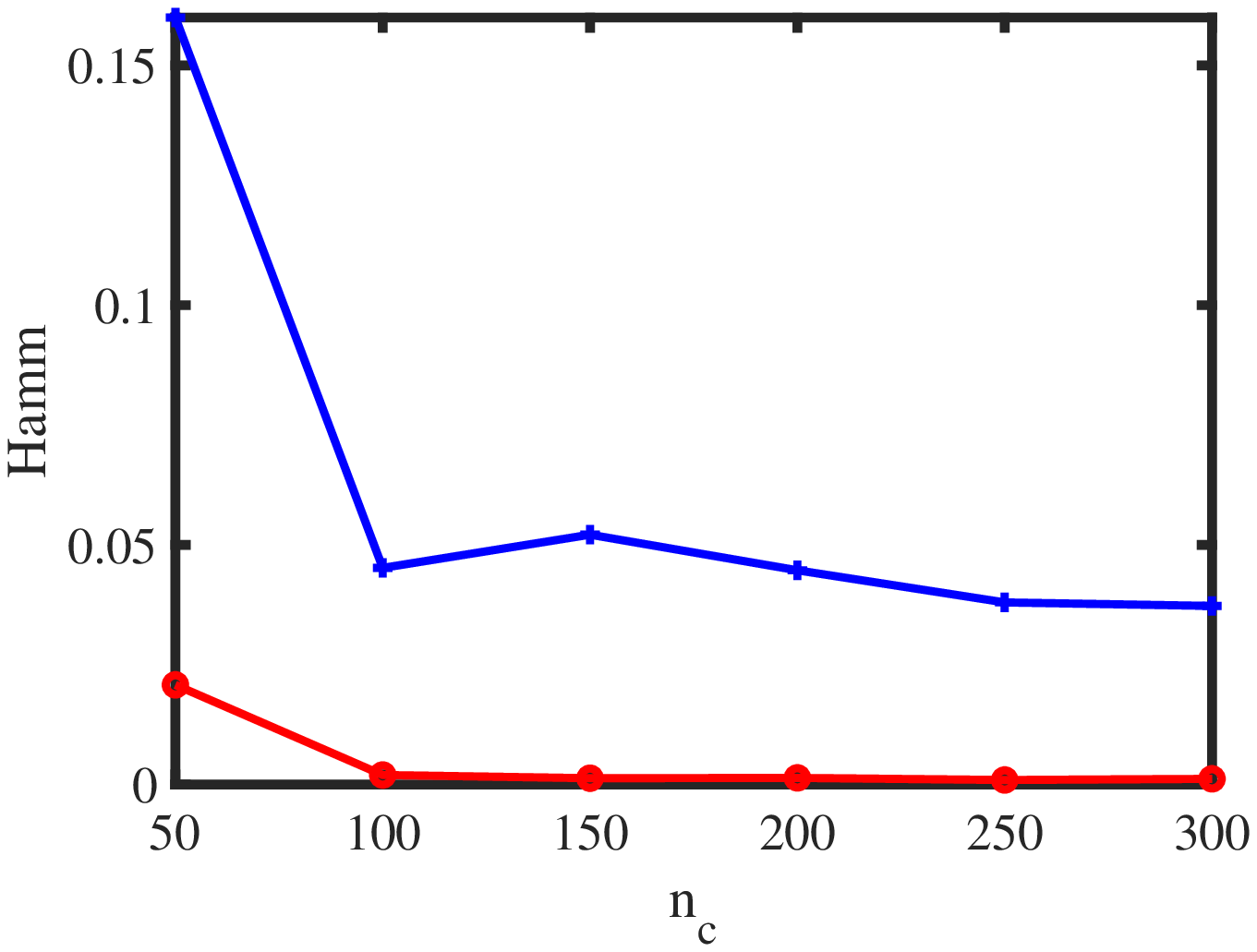}}
\subfigure[Changing $\rho$ under ONM: MHamm.]{\includegraphics[width=0.37\textwidth]{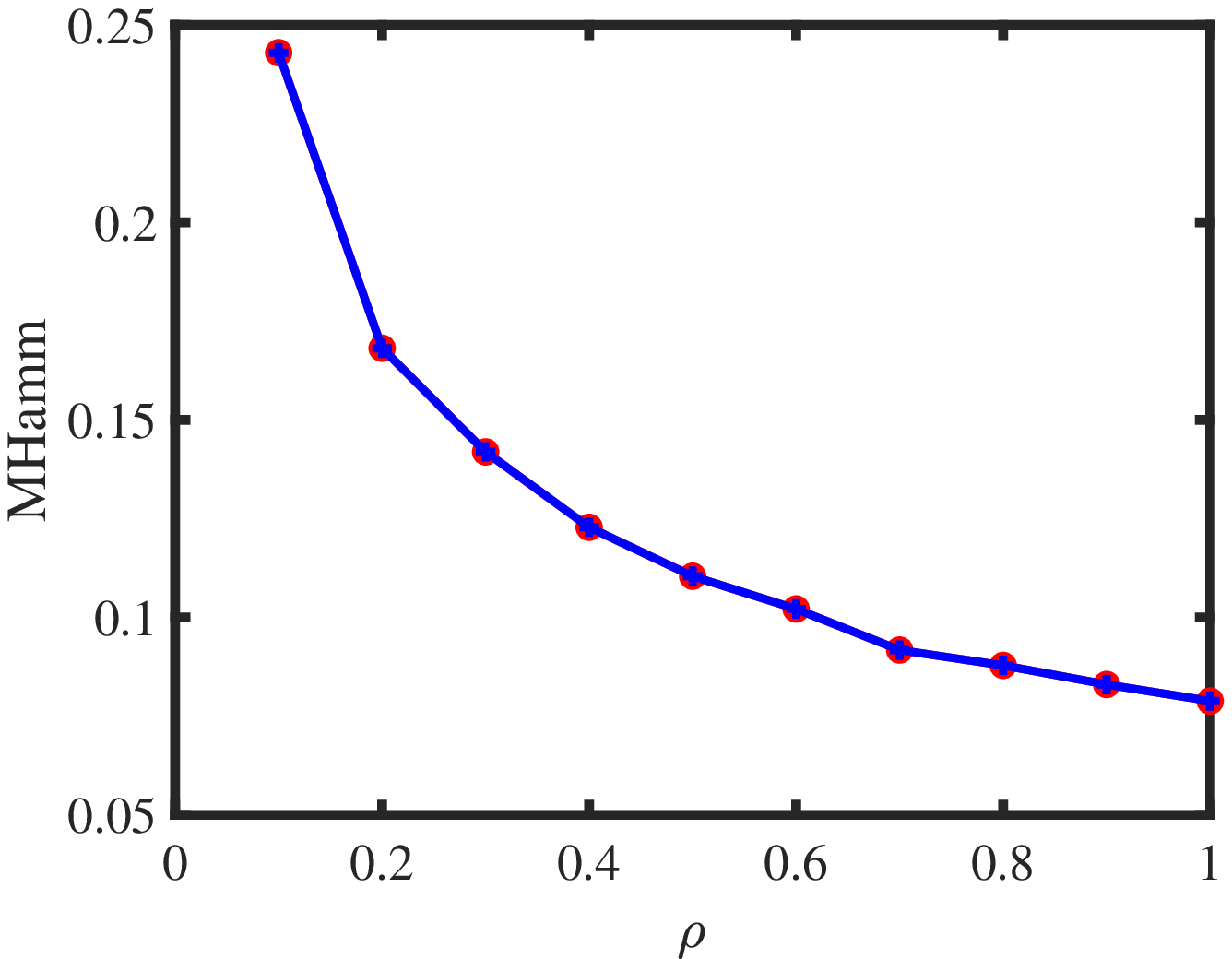}}
\subfigure[Changing $\rho$ under ONM: Hamm.]{\includegraphics[width=0.37\textwidth]{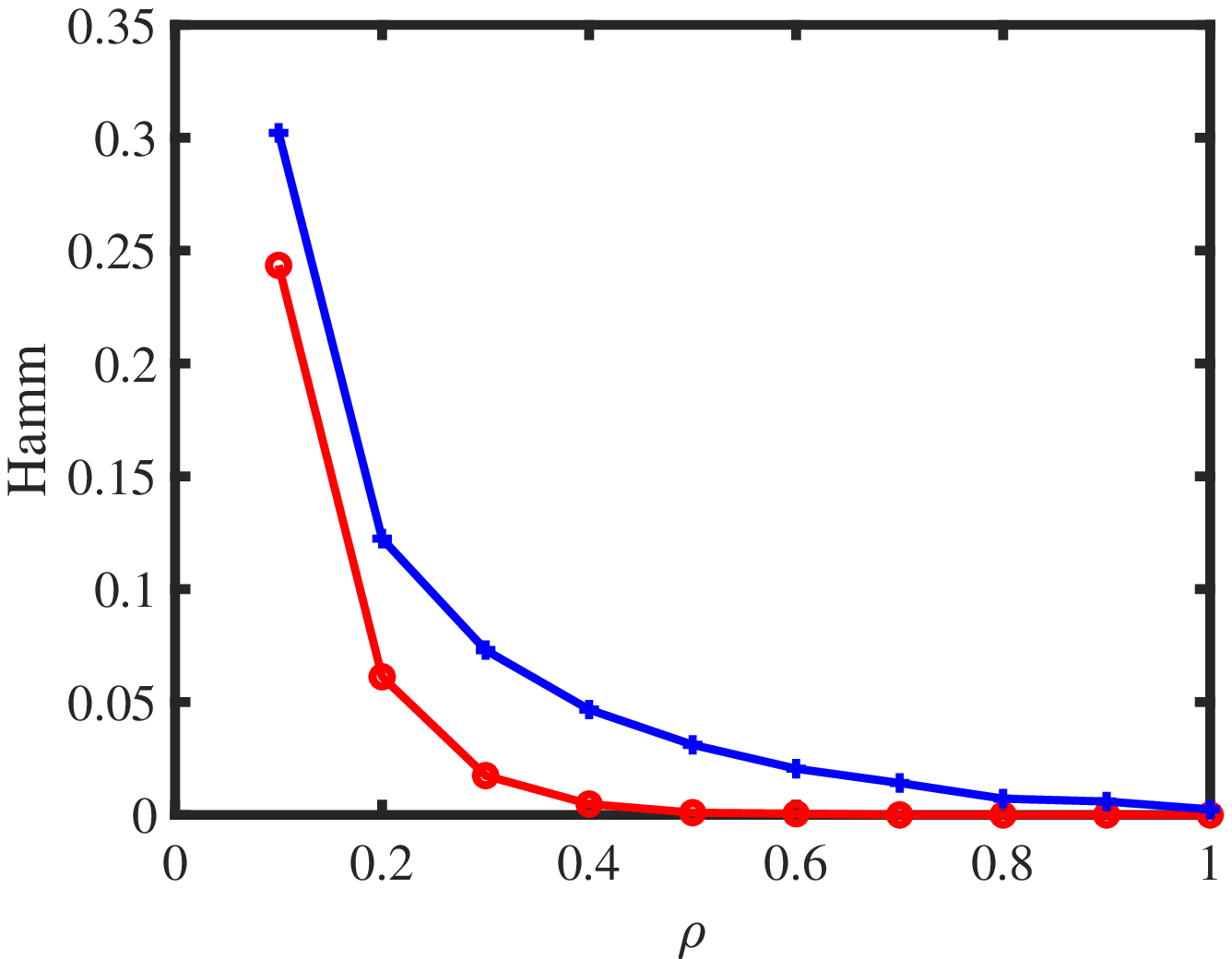}}
\subfigure[Changing $z_{c}$ under ODCNM: MHamm.]{\includegraphics[width=0.37\textwidth]{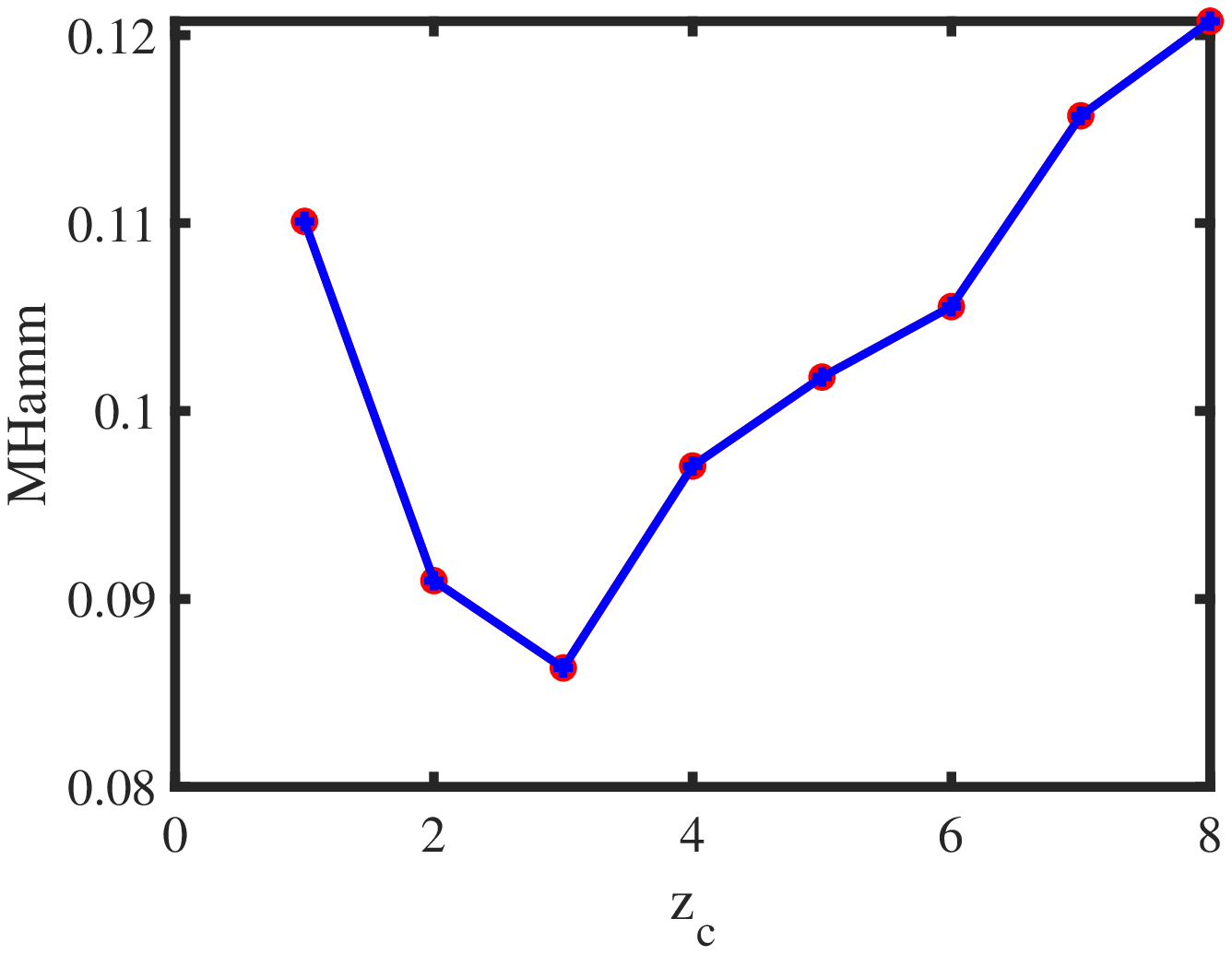}}
\subfigure[Changing $z_{c}$ under ODCNM: Hamm.]{\includegraphics[width=0.37\textwidth]{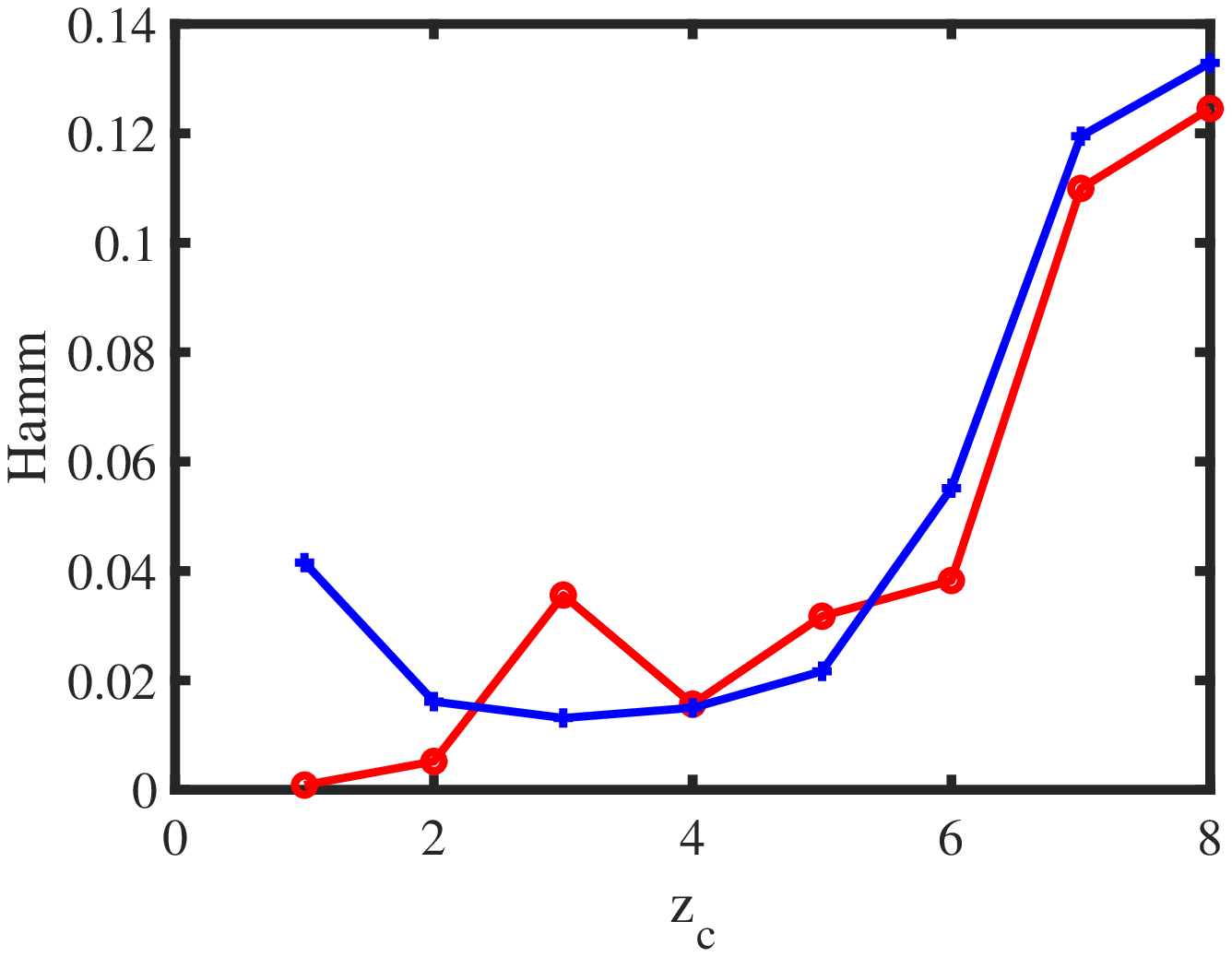}}
\subfigure[Changing $\rho$ under ODCNM: MHamm.]{\includegraphics[width=0.37\textwidth]{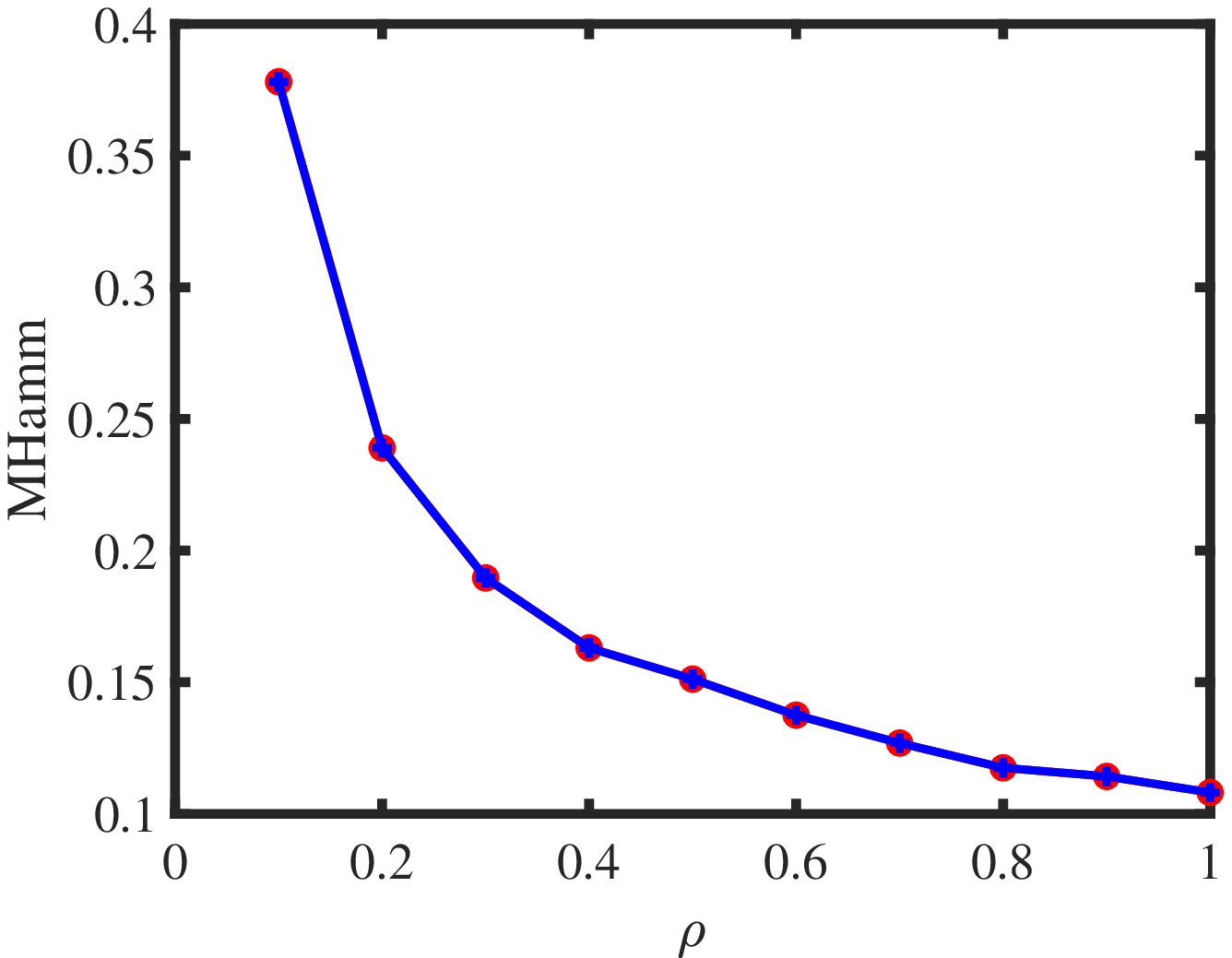}}
\subfigure[Changing $\rho$ under ODCNM: Hamm.]{\includegraphics[width=0.37\textwidth]{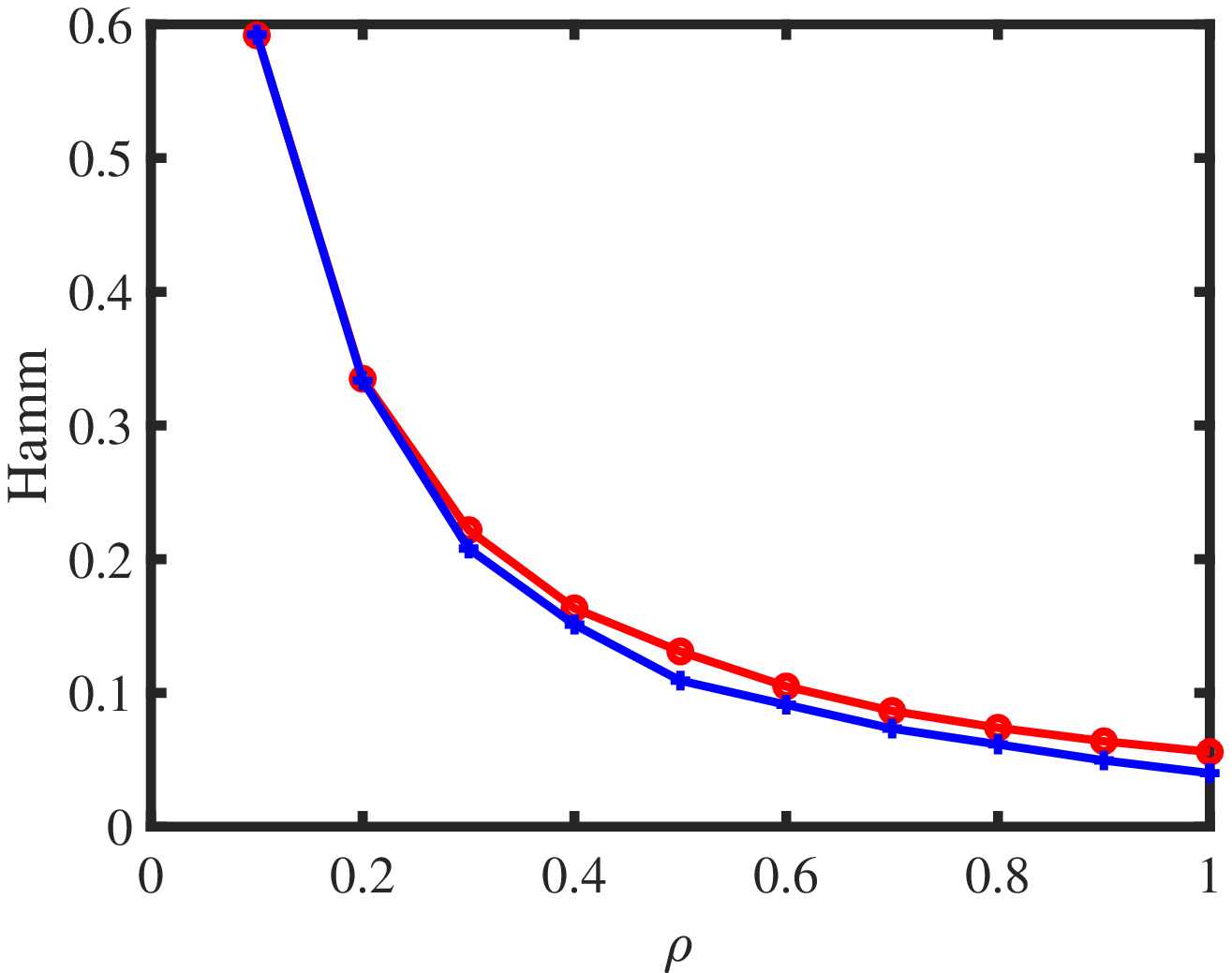}}
\caption{Estimation errors of ONA and ODCNA.}
\label{EX}
\end{figure}

\texttt{Experiment 2: Changing $\rho$ under ONM.} $P$ is set same as Experiment 1, and we let $\rho$ range in $\{0.1,0.2,\ldots,1\}$ to study the influence of $\rho$ on performances of ONA and ODCNA under ONM. The results are displayed in panels (c) and (d) of Figure \ref{EX}. From the results, we see that both methods perform better as $\rho$ increases since a larger $\rho$ gives more edges generated in a directed network. Meanwhile, the total run-time for this experiment is roughly 136 seconds.

\texttt{Experiment 3: Change $z_{c}$ under ODCNM.} $P$ is set same as Experiment 1. Let $z_{c}$ range in $\{1,2,\ldots,8\}$. Increasing $z_{c}$ decreases edges generated under ODCNM.  Panels (e) and (f) in Figure \ref{EX} display simulation results of this experiment. The results show that, generally, increasing the variability of node degrees makes it harder to detect node memberships for both ONA and ODCNA. Though ODCNA is designed under ODCNM, it holds similar performances as ONA for directed networks in which column nodes have various degrees in this experiment, and this is consistent with our theoretical findings in Corollaries \ref{AddConditions} and \ref{AddConditionsODCNM}. Meanwhile, the total run-time for this experiment is around 131 seconds.

\texttt{Experiment 4: Change $\rho$ under ODCNM.} Set $z_{c}=3$, $P$ is set same as Experiment 1, and let $\rho$ range in $\{0.1, 0.2,\ldots, 1\}$ under ODCNM. Panels (g) and (h) in Figure \ref{EX} displays simulation results of this experiment. The performances of the two proposed methods are similar as that of Experiment 2. Meanwhile, the total run-time for this experiment is around 221 seconds.
\section{Discussions}\label{sec7}
In this paper, we introduced overlapping and nonoverlapping models and its extension by considering degree heterogeneity. The models can model directed network with $K_{r}$ row communities and $K_{c}$ column communities, in which row node can belong to multiple row communities while column node only belong to one of the column communities. The proposed models are identifiable when $K_{r}\leq K_{c}$ and some other popular constraints on the connectivity matrix and membership matrices. For comparison, modeling directed network in which row nodes have overlapping property while column nodes do not with $K_{r}>K_{c}$ is unidentifiable. Meanwhile, since previous works found that modeling directed networks in which both row and column nodes have overlapping property with $K_{r}\neq K_{c}$ is unidentifiable, our identifiable ONM and ODCNM as well as the DCONM in Appendix \ref{DCOOOOOONM} supply a gap in modeling overlapping directed networks when $K_{r}\neq K_{c}$.  Theses models provide exploratory tools for studying community structure in directed networks with one side is overlapping while another side is nonoverlapping. Two spectral algorithms are designed to fit ONM and ODCNM. We also showed estimation consistency under mild conditions for our methods. Especially, when ODCNM reduces to ONM, our theoretical results under ODCNM are consistent with those under ONM. But perhaps the main limitation of the models is that the $K_{r}$ and $K_{c}$ in the directed network are assumed given, and such limitation also holds for the ScBM and DCScBM of \cite{DISIM}. In most community problems, the number of row community and the number of column community are unknown, therefore a complete calculation  and theoretical study require not only the algorithms and their theoretically consistent estimations described in this paper but also a method for estimating $K_{r}$ and $K_{c}$. We leave studies of this problem to our future work.
\appendix
\section{Successive Projection algorithm}
Algorithm \ref{alg:SP} is the Successive Projection algorithm.
	\begin{algorithm}
		\caption{\textbf{Successive Projection (SP)} \citep{gillis2015semidefinite}}
		\label{alg:SP}
		\begin{algorithmic}[1]
			\Require Near-separable matrix $Y_{sp}=S_{sp}M_{sp}+Z_{sp}\in\mathbb{R}^{m\times n}_{+}$ , where $S_{sp}, M_{sp}$ should satisfy Assumption 1 \cite{gillis2015semidefinite}, the number $r$ of columns to be extracted.
			\Ensure Set of indices $\mathcal{K}$ such that $Y_{sp}(\mathcal{K},:)\approx S$ (up to permutation)
			\State Let $R=Y_{sp}, \mathcal{K}=\{\}, k=1$.
			\State \textbf{While} $R\neq 0$ and $k\leq r$ \textbf{do}
			\State ~~~~~~~$k_{*}=\mathrm{argmax}_{k}\|R(k,:)\|_{F}$.
			\State ~~~~~~$u_{k}=R(k_{*},:)$.
			\State ~~~~~~$R\leftarrow (I-\frac{u_{k}u'_{k}}{\|u_{k}\|^{2}_{F}})R$.
			\State ~~~~~~$\mathcal{K}=\mathcal{K}\cup \{k_{*}\}$.
			\State ~~~~~~k=k+1.
			\State \textbf{end while}
		\end{algorithmic}
	\end{algorithm}
\section{Proofs under ONM}
\subsection{Proof of Proposition \ref{id}}
\begin{proof}
By Lemma \ref{RK}, let $U_{r}\Lambda U'_{c}$ be the compact SVD of $\Omega$ such that $\Omega=U_{r}\Lambda U'_{c}$, since $\Omega=\Pi_{r}P\Pi'_{c}=\check{\Pi}_{r}\check{P}\check{\Pi}'_{c}$, we have $\Omega(\mathcal{I}_{r}, \mathcal{I}_{c})=P=\check{P}$, which gives $P=\check{P}$. By Lemma \ref{RK}, since $U_{r}=\Pi_{r}U_{r}(\mathcal{I}_{r},:)=\check{\Pi}_{r}U_{r}(\mathcal{I}_{r},:)$, we have $\Pi_{r}=\check{\Pi}_{r}$ where we have used the fact that the inverse of $U_{r}(\mathcal{I}_{r},:)$ exists. Since $\Omega=\Pi_{r}P\Pi'_{c}=\check{\Pi}_{r}\check{P}\check{\Pi}'_{c}=\Pi_{r} P\check{\Pi}'_{c}$, we have $\Pi_{r}P\Pi'_{c}=\Pi_{r} P\check{\Pi}'_{c}$. By Lemma 7 of \cite{qing2021DiMMSB}, we have $P\Pi'_{c}=P\check{\Pi}'_{c}$, i.e., $\Pi_{c}X=\check{\Pi}_{c}X$ where we set $X=P'\in \mathbb{R}^{K_{c}\times K_{r}}$. Let $\check{\ell}$ be the $n_{c}\times 1$ vector of column nodes labels obtained from $\check{\Pi}_{c}$. For $i_{c}\in[n_{c}], k\in[K_{r}]$, from $\Pi_{c}X=\check{\Pi}_{c}X$, we have $(\Pi_{c}X)(i_{c},k)=\Pi_{c}(i_{c},:)X(:,k)=X(\ell(i_{c}),k)=X(\check{\ell}(i_{c}),k)$, which means that we must have $\ell(i_{c})=\check{\ell}(i_{c})$ for all $i_{c}\in[n_{c}]$, i.e., $\ell=\check{\ell}$ and $\Pi_{c}=\check{\Pi}_{c}$. Note that, for the special case $K_{r}=K_{c}=K$, $\Pi_{c}=\check{\Pi}_{c}$ can be obtained easily: since $P\Pi'_{c}=P\check{\Pi}'_{c}$ and $P\in\mathbb{R}^{K\times K}$ is assumed to be full rank, we have $\Pi_{c}=\check{\Pi}_{c}$. Thus the proposition holds.
\end{proof}
\subsection{Proof of Lemma \ref{RK}}
\begin{proof}
For $U_{r}$, since $\Omega=U_{r}\Lambda U'_{c}$ and $U'_{c}U_{c}=I_{K_{r}}$, we have $U_{r}=\Omega U_{c}\Lambda^{-1}$. Recall that $\Omega=\Pi_{r}P\Pi'_{c}$, we have $U_{r}=\Pi_{r}P\Pi'_{c}U_{c}\Lambda^{-1}=\Pi_{r}B_{r}$, where we set $B_{r}=P\Pi'_{c}U_{c}\Lambda^{-1}$. Since $U_{r}(\mathcal{I}_{r},:)=\Pi_{r}(\mathcal{I}_{r},:)B_{r}=B_{r}$, we have $B_{r}=U_{r}(\mathcal{I}_{r},:)$. For $i_{r}\in[n_{r}]$, $U_{r}(i_{r},:)=e'_{i_{r}}\Pi_{r}B_{r}=\Pi_{r}(i_{r},:)B_{r}$, so sure we have $U_{r}(i_{r},:)=U_{r}(\bar{i}_{r},:)$ when $\Pi_{r}(i_{r},:)=\Pi_{r}(\bar{i}_{r},:)$.

For $U_{c}$, follow similar analysis as for $U_{r}$, we have $U_{c}=\Pi_{c}B_{c}$, where $B_{c}=P'\Pi'_{r}U_{r}\Lambda^{-1}$. Note that $B_{c}\in \mathbb{R}^{K_{c}*K_{r}}$. Sure, $U_{c}(i_{c},:)=U_{c}(\bar{i}_{c},:)$ when $\ell(i_{c})=\ell(\bar{i}_{c})$ for $i_{c},\bar{i}_{c}\in[n_{c}]$.

Now, we focus on the case when $K_{r}=K_{c}=K$. For this case, since $B_{c}\in \mathbb{R}^{K_{c}*K_{r}}$, $B_{c}$ is full rank when $K_{r}=K_{c}$. Since $I_{K_{r}}=I_{K}=U'_{c}U_{c}=B'_{c}\Pi'_{c}\Pi_{c}B_{c}$, we have $\Pi'_{c}\Pi_{c}=(B_{c}B'_{c})^{-1}$. Since $\Pi'_{c}\Pi_{c}=\mathrm{diag}(n_{c,1}, n_{c,2}, \ldots, n_{c,K})$, we have $B_{c}B'_{c}=\mathrm{diag}(\frac{1}{n_{c,1}}, \frac{1}{n_{c,2}}, \ldots, \frac{1}{n_{c,K}})$. When $K_{r}=K_{c}=K$, we have $B_{c}(k,:)B'_{c}(l,:)=0$ for any $k\neq l$ and $k,l\in [K]$. Then ,we have $B_{c}B'_{c}=\mathrm{diag}(\|B_{c}(1,:)\|^{2}_{F},\|B_{c}(2,:)\|^{2}_{F},\ldots, \|B_{c}(K,:)\|^{2}_{F})=\mathrm{diag}(\frac{1}{n_{c,1}}, \frac{1}{n_{c,2}}, \ldots, \frac{1}{n_{c,K}})$ and the lemma follows.

Note that when $K_{r}<K_{c}$, since $B_{c}$ is not full rank now, we can not obtain $\Pi'_{c}\Pi_{c}=(B_{c}B'_{c})^{-1}$ from $I_{K_{r}}=B'_{c}\Pi'_{c}\Pi_{c}B_{c}$. Therefore, when $K_{r}<K_{c}$, the equality $\|B_{c}(k,:)-B_{c}(l,:)\|_{F}=\sqrt{\frac{1}{n_{c,k}}+\frac{1}{n_{c,l}}}$ does not hold for any $k\neq l$. And we can only know that $U_{c}$ has $K_{c}$ distinct rows when $K_{r}<K_{c}$, but have no knowledge about the minimum distance between any two distinct rows of $U_{c}$.
\end{proof}
\subsection{Proof of Theorem \ref{Main}}
\begin{proof}
For row nodes, when conditions in Lemma \ref{rowwiseerror} hold,  by Theorem 2 of \cite{qing2021DiMMSB}, with probability at least $1-o((n_{r}+n_{c})^{-\alpha})$ for  any $\alpha>0$, there exists a permutation matrix  $\mathcal{P}_{r}$ such that, for $i_{r}\in[n_{r}]$, we have
\begin{align*}	\|e'_{i_{r}}(\hat{\Pi}_{r}-\Pi_{r}\mathcal{P}_{r})\|_{1}=O(\varpi\kappa(\Pi'_{r}\Pi_{r})K_{r}\sqrt{\lambda_{1}(\Pi'_{r}\Pi_{r})}).
\end{align*}
Next, we focus on column nodes. By the proof of Lemma 2.3 of \cite{qing2021consistency}, there exists an orthogonal matrix $\hat{O}$ such that
\begin{align}\label{UC1}
\|\hat{U}_{c}\hat{O}-U_{c}\|_{F}\leq \frac{2\sqrt{2K_{r}}\|A-\Omega\|}{\sqrt{\lambda_{K_{r}}(\Omega'\Omega)}}.
\end{align}
Under $ONM_{n_{r},n_{c}}(K_{r}, K_{c}, P,\Pi_{r}, \Pi_{c})$, by Lemma 10 of \cite{qing2021DiMMSB}, we have
\begin{align}\label{UC2}
\sqrt{\lambda_{K_{r}}(\Omega'\Omega)}\geq \rho \sigma_{K_{r}}(\tilde{P})\sigma_{K_{r}}(\Pi_{r})\sigma_{K_{r}}(\Pi_{c}).
\end{align}
Since all column nodes are pure, $\sigma_{K_{r}}(\Pi_{c})=\sqrt{n_{c,K_{r}}}$. By Lemma 3 of \cite{qing2021DiMMSB}, when Assumption (\ref{a1}) holds, with probability at least $1-o((n_{r}+n_{c})^{-\alpha})$, we have
\begin{align}\label{UC3}
\|A-\Omega\|=O(\sqrt{\rho \mathrm{max}(n_{r},n_{c})\mathrm{log}(n_{r}+n_{c})}).
\end{align}
Substitute the two bounds in Eqs (\ref{UC2}) and (\ref{UC3}) into Eq (\ref{UC1}), we have
\begin{align}\label{UC4}
\|\hat{U}_{c}\hat{O}-U_{c}\|_{F}\leq C\frac{\sqrt{K_{r}\mathrm{max}(n_{r},n_{c})\mathrm{log}(n_{r}+n_{c})}}{\sigma_{K_{r}}(\tilde{P})\sqrt{\rho}\sigma_{K_{r}}(\Pi_{r})\sqrt{n_{c,K_{r}}}}.
\end{align}
Let $\varsigma>0$ be a small quantity, by Lemma 2 in \cite{joseph2016impact}, if
\begin{align}\label{holdcONM}
\frac{\sqrt{K_{c}}}{\varsigma}\|U_{c}-\hat{U}_{c}\hat{O}\|_{F}(\frac{1}{\sqrt{n_{c,k}}}+\frac{1}{\sqrt{n_{c,l}}})\leq \|B_{c}(k,:)-B_{c}(l,:)\|_{F}, \mathrm{~for~each~}1\leq k\neq l\leq K_{c},
\end{align}
then the clustering error $\hat{f}_{c}=O(\varsigma^{2})$. Recall that we set $\delta_{c}=\mathrm{min}_{k\neq l}\|B_{c}(k,:)-B_{c}(l,:)\|_{F}$ to measure the minimum center separation of $B_{c}$. Setting $\varsigma=\frac{2}{\delta_{c}}\sqrt{\frac{K_{c}}{n_{c,\mathrm{min}}}}\|U_{c}-\hat{U}_{c}\hat{O}\|_{F}$ makes Eq (\ref{holdcONM}) hold for all $1\leq k\neq l\leq K_{c}$. Then we have $\hat{f}_{c}=O(\varsigma^{2})=O(\frac{K_{c}\|U_{c}-\hat{U}_{c}\hat{O}\|^{2}_{F}}{\delta^{2}_{c}n_{c,\mathrm{min}}})$. By Eq (\ref{UC4}), we have
\begin{align*}
\hat{f}_{c}=O(\frac{K_{r}K_{c}\mathrm{max}(n_{r},n_{c})\mathrm{log}(n_{r}+n_{c})}{\sigma^{2}_{K_{r}}(\tilde{P})\rho \delta^{2}_{c}\sigma^{2}_{K_{r}}(\Pi_{r})n_{c,K_{r}}n_{c,\mathrm{min}}}).
\end{align*}
Especially, when $K_{r}=K_{c}=K$, $\delta_{c}\geq \sqrt{\frac{2}{n_{c,\mathrm{max}}}}$ under $ONM_{n_{r},n_{c}}(K_{r},K_{c}, P, \Pi_{r}, \Pi_{c})$ by Lemma \ref{RK}. When $K_{r}=K_{c}=K$, we have
\begin{align*}
\hat{f}_{c}=O(\frac{K^{2}\mathrm{max}(n_{r},n_{c})n_{c,\mathrm{max}}\mathrm{log}(n_{r}+n_{c})}{\sigma^{2}_{K}(\tilde{P})\rho\sigma^{2}_{K}(\Pi_{r})n^{2}_{c,\mathrm{min}}}).
\end{align*}
\end{proof}
\subsection{Proof of Corollary \ref{AddConditions}}
\begin{proof}
For row nodes, under conditions of Corollary \ref{AddConditions}, we have
\begin{align*}
\mathrm{max}_{i_{r}\in[n_{r}]}\|e'_{i_{r}}(\hat{\Pi}_{r}-\Pi_{r}\mathcal{P}_{r})\|_{1}=O(\varpi K_{r}\sqrt{\frac{n_{r}}{K_{r}}})=O(\varpi\sqrt{Kn_{r}}).
\end{align*}
Under  conditions of Corollary \ref{AddConditions}, $\kappa(\Omega)=O(1)$ and $\mu\leq C$ for some $C>0$ by the proof of Corollary 1 \cite{qing2021DiMMSB}. Then, by Lemma \ref{rowwiseerror}, we have
\begin{align*}
\varpi&=O(\frac{\sqrt{K}(\kappa(\Omega)\sqrt{\frac{\mathrm{max}(n_{r},n_{c})\mu}{\mathrm{min}(n_{r},n_{c})}}+\sqrt{\mathrm{log}(n_{r}+n_{c})})}{\sqrt{\rho}\sigma_{K}(\tilde{P})\sigma_{K}(\Pi_{r})\sqrt{n_{c,K_{r}}}})=O(\frac{\sqrt{K}(\sqrt{\frac{C\mathrm{max}(n_{r},n_{c})}{\mathrm{min}(n_{r},n_{c})}}+\sqrt{\mathrm{log}(n_{r}+n_{c})})}{\sqrt{\rho}\sigma_{K}(\tilde{P})\sigma_{K}(\Pi_{r})\sqrt{n_{c,\mathrm{min}}}})\\
&=O(\frac{K^{1.5}(\sqrt{\frac{C\mathrm{max}(n_{r},n_{c})}{\mathrm{min}(n_{r},n_{c})}}+\sqrt{\mathrm{log}(n_{r}+n_{c})})}{\sigma_{K}(\tilde{P})\sqrt{\rho n_{r}n_{c}}}),
\end{align*}
which gives that
\begin{align*}
\mathrm{max}_{i_{r}\in[n_{r}]}\|e'_{i_{r}}(\hat{\Pi}_{r}-\Pi_{r}\mathcal{P}_{r})\|_{1}=O(\frac{K^{2}(\sqrt{\frac{C\mathrm{max}(n_{r},n_{c})}{\mathrm{min}(n_{r},n_{c})}}+\sqrt{\mathrm{log}(n_{r}+n_{c})})}{\sigma_{K}(\tilde{P})\sqrt{\rho
n_{c}}}).
\end{align*}
Note that, when $K_{r}<K_{c}$, we can not draw a conclusion that $\mu\leq C$. Because, when $K_{r}<K_{c}$, the inverse of $B_{c}B'_{c}$ does not exist since $B_{c}\in\mathbb{R}^{K_{c}\times K_{r}}$. Therefore, Lemma 8 of \cite{qing2021DiMMSB} does not hold, and we can not obtain  the upper bound of $\|U_{c}\|_{2\rightarrow\infty}$, causing the impossibility of obtaining the upper bound of $\mu$, and this is the reason that we only consider the case when $K_{r}=K_{c}$ for row nodes here.

For column nodes, under conditions of Corollary \ref{AddConditions}, we have
  \begin{align*}
    \hat{f}_{c}&=O(\frac{K_{r}K_{c}\mathrm{max}(n_{r},n_{c})\mathrm{log}(n_{r}+n_{c})}{\sigma^{2}_{K_{r}}(\tilde{P})\rho \delta^{2}_{c}\sigma^{2}_{K_{r}}(\Pi_{r})n_{c,K_{r}}n_{c,\mathrm{min}}})=O(\frac{K_{r}K_{c}\mathrm{max}(n_{r},n_{c})\mathrm{log}(n_{r}+n_{c})}{\sigma^{2}_{K_{r}}(\tilde{P})\rho \delta^{2}_{c}(n_{r}/K_{r})(n_{c}/K_{c})(n_{c}/K_{c})})\\
    &=O(\frac{K^{2}_{r}K^{3}_{c}\mathrm{max}(n_{r},n_{c})\mathrm{log}(n_{r}+n_{c})}{\sigma^{2}_{K_{r}}(\tilde{P})\rho \delta^{2}_{c}n_{r}n^{2}_{c}}).
  \end{align*}
For the special case $K_{r}=K_{c}=K$, since $\frac{n_{c,\mathrm{max}}}{n_{c,\mathrm{min}}}=O(1)$ when $n_{c,\mathrm{min}}=O(\frac{n_{c}}{K})$, we have
\begin{align*}
\hat{f}_{c}=O(\frac{K^{4}\mathrm{max}(n_{r},n_{c})\mathrm{log}(n_{r}+n_{c})}{\sigma^{2}_{K}(\tilde{P})\rho n_{r}n_{c}}).
\end{align*}
When $n_{r}=O(n), n_{c}=O(n), K_{r}=O(1)$ and $K_{c}=O(1)$, the corollary follows immediately by basic algebra.
\end{proof}
\section{Proofs under ODCNM}
\subsection{Proof of Proposition \ref{idODCNM}}
\begin{proof}
Since $\Omega=\Pi_{r}P\Pi'_{c}\Theta_{c}=\check{\Pi}_{r}\check{P}\check{\Pi}'_{c}\check{\Theta}_{c}=U_{r}\Lambda U'_{c}$, we have $U_{r}=\Pi_{r}U_{r}(\mathcal{I}_{r},:)=\check{\Pi}_{r}U_{r}(\mathcal{I}_{r},:)$ by Lemma \ref{RKODCNM}, which gives that $\Pi_{r}=\check{\Pi}_{r}$. Since $U_{c,*}=\Pi_{c}B_{c}=\Pi_{c}U_{c,*}(\mathcal{I}_{c},:)=\check{\Pi}_{c}U_{c,*}(\mathcal{I}_{c},:)$ by Lemma \ref{RKODCNM}, we have $\Pi_{c}=\check{\Pi}_{c}$.
\end{proof}
\subsection{Proof of Lemma \ref{RKODCNM}}
\begin{proof}
\begin{itemize}
  \item For $U_{r}$: since $\Omega=U_{r}\Lambda U'_{c}$ and $U'_{c}U_{c}=I_{K_{r}}$, we have $U_{r}=\Omega U_{c}\Lambda^{-1}$. Recall that $\Omega=\Pi_{r}P\Pi'_{c}\Theta_{c}$ under ODCNM, we have $U_{r}=\Pi_{r}P\Pi'_{c}\Theta_{c}U_{c}\Lambda^{-1}=\Pi_{r}B_{r}$, where $B_{r}=P\Pi'_{c}\Theta_{c}U_{c}\Lambda^{-1}$. Sure, $U_{r}(i_{r},:)=U_{r}(\bar{i}_{r},:)$ holds when $\Pi_{r}(i_{r},:)=\Pi_{r}(\bar{i}_{r},:)$ for $i_{r},\bar{i}_{r}\in[n_{r}]$.
  \item For $U_{c}$: let $D_{c}$ be a $K_{c}\times K_{c}$ diagonal matrix such that $D_{c}(k,k)=\frac{\|\Theta_{c}\Pi_{c}(:,k)\|_{F}}{\|\theta_{c}\|_{F}}$ for $k\in[K_{c}]$. Let $\Gamma_{c}$ be an $n_{c}\times K_{c}$ matrix such that $\Gamma_{c}(:,k)=\frac{\Theta_{c}\Pi_{c}(:,k)}{\|\Theta_{c}\Pi_{c}(:,k)\|_{F}}$ for $k\in[K_{c}]$. For such $D_{c}$ and $\Gamma_{c}$, we have $\Gamma'_{c}\Gamma_{c}=I_{K_{c}}$ and $\Omega=\Pi_{r}P\|\theta_{c}\|_{F}D_{c}\Gamma'_{c}$, i.e., $\Theta_{c}\Pi_{c}=\|\theta_{c}\|_{F}\Gamma_{c}D_{c}$.

      Since $\Omega=U_{r}\Lambda U'_{c}$ and $U'_{r}U_{r}=I_{K_{r}}$, we have $U_{c}=\Theta_{c}\Pi_{c}P'\Pi'_{r}U_{r}\Lambda^{-1}$. Since $\Theta_{c}\Pi_{c}=\|\theta_{c}\|_{F}\Gamma_{c}D_{c}$, we have $U_{c}=\Gamma_{c}\|\theta_{c}\|_{F}D_{c}P'\Pi'_{r}U_{r}\Lambda^{-1}=\Gamma_{c}V_{c}$, where we set $V_{c}=\|\theta_{c}\|_{F}D_{c}P'\Pi'_{r}U_{r}\Lambda^{-1}\in \mathbb{R}^{K_{c}\times K_{r}}$. Note that since $U'_{c}U_{c}=I_{K_{r}}=V'_{c}\Gamma'_{c}\Gamma_{c}V_{c}=V'_{c}V_{c}$, we have $V'_{c}V_{c}=I_{K_{r}}$. Now, for $i_{c}\in[n_{c}], k\in[K_{r}]$, we have
\begin{align*}
U_{c}(i_{c},k)&=e'_{i_{c}}U_{c}e_{k}=e'_{i_{c}}\Gamma_{c}V_{c}e_{k}=\Gamma_{c}(i_{c},:)V_{c}e_{k}\\
&=\theta_{c}(i_{c})[\frac{\Pi_{c}(i_{c},1)}{\|\Theta_{c}\Pi_{c}(:,1)\|_{F}}~~\frac{\Pi_{c}(i_{c},2)}{\|\Theta_{c}\Pi_{c}(:,2)\|_{F}}~~\ldots~~\frac{\Pi_{c}(i_{c},K_{c})}{\|\Theta_{c}\Pi_{c}(:,K_{c})\|_{F}}]V_{c}e_{k}\\
&=\frac{\theta_{c}(i_{c})}{\|\Theta_{c}\Pi_{c}(:,\ell(i_{c}))\|_{F}}V_{c}(\ell(i_{c}),k),
\end{align*}
which gives that
\begin{align*}
U_{c}(i_{c},:)=\frac{\theta_{c}(i_{c})}{\|\Theta_{c}\Pi_{c}(:,\ell(i_{c})\|_{F}}[V_{c}(\ell(i_{c}),1)~~V_{c}(\ell(i_{c}),2)~~\ldots~~V_{c}(\ell(i_{c}),K_{r})]=\frac{\theta_{c}(i_{c})}{\|\Theta_{c}\Pi_{c}(:,\ell(i_{c})\|_{F}}V_{c}(\ell(i_{c}),:).
\end{align*}
Then we have
\begin{align}\label{Uex}
U_{c,*}(i_{c},:)=\frac{V_{c}(\ell(i_{c}),:)}{\|V_{c}(\ell(i_{c}),:)\|_{F}}.
\end{align}
Sure, we have $U_{c,*}(i_{c},:)=U_{c,*}(\bar{i}_{c},:)$ when $\ell(i_{c})=\ell(\bar{i}_{c})$ for $i_{c},\bar{i}_{c}\in[n_{c}]$. Let $B_{c}\in\mathbb{R}^{K_{c}\times K_{r}}$ such that $B_{c}(l,:)=\frac{V_{c}(l,:)}{\|V_{c}(l,:)\|_{F}}$ for $l\in[K_{c}]$. Eq (\ref{Uex}) gives $U_{c,*}=\Pi_{c}B_{c}$, which guarantees the existence of $B_{c}$.

Now we consider the case when $K_{r}=K_{c}=K$. Since $V_{c}\in\mathbb{R}^{K_{c}\times K_{r}}$ and $U_{c}=\Gamma_{c}V_{c}\in\mathbb{R}^{n_{c}\times K_{r}}$, we have $V_{c}\in\mathbb{R}^{K\times K}$ and $\mathrm{rank}(V_{c})=K$. Since $V'_{c}V_{c}=I_{K_{r}}$, we have $V'_{c}V_{c}=I_{K}$ when $K_{r}=K_{c}=K$. Then we have
\begin{align}\label{EqV}
V'_{c}V_{c}=I_{K}\Rightarrow V'_{c}V_{c}V'_{c}=V'_{c}\Rightarrow V'_{c}(V_{c}V'_{c}-I_{K})=0\overset{\mathrm{rank}(V_{c})=K}{\Rightarrow}V_{c}V'_{c}=I_{K}.
\end{align}
Since  $V_{c}V'_{c}=V'_{c}V_{c}=I_{K}$, we have $U_{c,*}(i_{c},:)=V_{c}(\ell(i_{c}),:)$ by Eq (\ref{Uex}), and $\|U_{c,*}(i_{c},:)-U_{c,*}(\bar{i}_{c},:)\|_{F}=\|V_{c}(\ell(i_{c}),:)-V_{c}(\ell(\bar{i}_{c}),:)\|_{F}=\sqrt{2}$ when $\ell(i_{c})\neq \ell(\bar{i}_{c})$ for $i_{c},\bar{i}_{c}\in[n_{c}]$, i.e., $\|B_{c}(k,:)-B_{c}(l,:)\|_{F}=\sqrt{2}$ for $k\neq l\in[K]$.

Note that, when $K_{r}<K_{c}$, since $\mathrm{rank}(V_{c})=K_{r}$ and $V_{c}\in\mathbb{R}^{K_{c}\times K_{r}}$, the inverse of $V_{c}$ does not exist, which causes that the last equality in Eq (\ref{EqV}) does not hold and $\|B_{c}(k,:)-B_{c}(\ell,:)\|\neq \sqrt{2}$ for all $k\neq l\in[K_{c}]$.
\end{itemize}
\end{proof}
\subsection{Proof of Theorem \ref{MainODCNM}}
\begin{proof}
For row nodes, when conditions in Lemma \ref{rowwiseerrorODCNM} hold, by Theorem 2 of \cite{qing2021DiMMSB}, we have
\begin{align*}	\mathrm{max}_{i_{r}\in[n_{r}]}\|e'_{i_{r}}(\hat{\Pi}_{r}-\Pi_{r}\mathcal{P}_{r})\|_{1}=O(\varpi\kappa(\Pi'_{r}\Pi_{r})K_{r}\sqrt{\lambda_{1}(\Pi'_{r}\Pi_{r})}).
\end{align*}
Next, we focus on column nodes. By the proof of Lemma 2.3 of \cite{qing2021consistency}, there exists an orthogonal matrix $\hat{O}$ such that
\begin{align}\label{UC1ODCNM}
\|\hat{U}_{c}\hat{O}-U_{c}\|_{F}\leq \frac{2\sqrt{2K_{r}}\|A-\Omega\|}{\sqrt{\lambda_{K_{r}}(\Omega'\Omega)}}.
\end{align}
Under $ODCNM_{n_{r},n_{c}}(K_{r}, K_{c}, P,\Pi_{r}, \Pi_{c},\Theta_{c})$, by Lemma 4 of \cite{qing2021DiDCMM}, we have
\begin{align}\label{UC2ODCNM}
\sqrt{\lambda_{K_{r}}(\Omega'\Omega)}\geq \theta_{c,\mathrm{min}}\sigma_{K_{r}}(P)\sigma_{K_{r}}(\Pi_{r})\sqrt{n_{c,K_{r}}}.
\end{align}
By Lemma 4.2 of \cite{qing2021DiDCMM}, when Assumption (\ref{a2}) holds, with probability at least $1-o((n_{r}+n_{c})^{-\alpha})$, we have
\begin{align}\label{UC3ODCNM}
\|A-\Omega\|=O(\sqrt{\mathrm{max}(\theta_{c,\mathrm{max}}n_{r},\|\theta_{c}\|_{1})\mathrm{log}(n_{r}+n_{c})}).
\end{align}
Substitute the two bounds in Eqs (\ref{UC2ODCNM}) and (\ref{UC3ODCNM}) into Eq (\ref{UC1ODCNM}), we have
\begin{align}\label{UC4ODCNM}
\|\hat{U}_{c}\hat{O}-U_{c}\|_{F}\leq C\frac{\sqrt{K_{r}\mathrm{max}(
\theta_{c,\mathrm{max}}n_{r},\|\theta_{c}\|_{1})\mathrm{log}(n_{r}+n_{c})}}{\sigma_{K_{r}}(P)\theta_{c,\mathrm{min}}\sigma_{K_{r}}(\Pi_{r})\sqrt{n_{c,K_{r}}}}.
\end{align}
For $i_{c}\in[n_{c}]$, by basic algebra, we have
\begin{align*}
\|\hat{U}_{c,*}(i_{c},:)\hat{O}-U_{c,*}(i_{c},:)\|_{F}\leq \frac{2\|\hat{U}_{c}(i_{c},:)\hat{O}-U_{c}(i_{c},:)\|_{F}}{\|U_{c}(i_{c},:)\|_{F}}.
\end{align*}
Set $m_{c}=\mathrm{min}_{1\leq i_{c}\leq n_{c}}\|U_{c}(i_{c},:)\|_{F}$, we have
\begin{align*}
\|\hat{U}_{c,*}\hat{O}-U_{c,*}\|_{F}=\sqrt{\sum_{i_{c}=1}^{n_{c}}\|\hat{U}_{c,*}(i_{c},:)\hat{O}-U_{c,*}(i_{c},:)\|^{2}_{F}}\leq \frac{2\|\hat{U}_{c}\hat{O}-U_{c}\|_{F}}{m_{c}}.
\end{align*}
Next, we provide lower bounds of $m_{c}$. By the proof of Lemma \ref{RKODCNM}, we have
\begin{align*}
\|U_{c}(i_{c},:)\|_{F}&=\|\frac{\theta_{c}(i_{c})}{\|\Theta_{c}\Pi_{c}(:,\ell(i_{c}))\|_{F}}V_{c}(\ell(i_{c}),:)\|_{F}=\frac{\theta_{c}(i_{c})}{\|\Theta_{c}\Pi_{c}(:,\ell(i_{c}))\|_{F}}\|V_{c}(\ell(i_{c}),:)\|_{F}\\
&=\frac{\theta_{c}(i_{c})}{\|\Theta_{c}\Pi_{c}(:,\ell(i_{c}))\|_{F}}\geq\frac{\theta_{c,\mathrm{min}}}{\theta_{c,\mathrm{max}}\sqrt{n_{c,\mathrm{max}}}}m_{V_{c}},
\end{align*}
where we set $m_{V_{c}}=\mathrm{min}_{k\in[K_{c}]}\|V_{c}(k,:)\|_{F}$. Note that when $K_{r}=K_{c}=K$, by the proof of Lemma \ref{RKODCNM}, we know that $V_{c}V'_{c}=I_{K}$, which gives that $\|V_{c}(k,:)\|_{F}=1$ for $k\in[K]$, i.e., $m_{V_{c}}=1$ when $K_{r}=K_{c}=K$. However, when $K_{r}<K_{c}$, it is challenge to obtain a positive lower bound of $m_{V_{c}}$.
Hence, we have $\frac{1}{m_{c}}\leq \frac{\theta_{c,\mathrm{max}}\sqrt{n_{c,\mathrm{max}}}}{\theta_{c,\mathrm{min}}m_{V_{c}}}$.
Then, by Eq (\ref{UC4ODCNM}), we have
\begin{align*}
\|\hat{U}_{c,*}\hat{O}-U_{c,*}\|_{F}=
O(\frac{\theta_{c,\mathrm{max}}\sqrt{K_{r}\mathrm{max}(
\theta_{c,\mathrm{max}}n_{r},\|\theta_{c}\|_{1})n_{c,\mathrm{max}}\mathrm{log}(n_{r}+n_{c})}}{\sigma_{K_{r}}(P)\theta^{2}_{c,\mathrm{min}}m_{V_{c}}\sigma_{K_{r}}(\Pi_{r})\sqrt{n_{c,K_{r}}}}).
\end{align*}

Let $\varsigma>0$ be a small quantity, by Lemma 2 in \cite{joseph2016impact}, if
\begin{align}\label{holdcODCNM}
\frac{\sqrt{K_{c}}}{\varsigma}\|U_{c,*}-\hat{U}_{c,*}\hat{O}\|_{F}(\frac{1}{\sqrt{n_{c,k}}}+\frac{1}{\sqrt{n_{c,l}}})\leq \|B_{c}(k,:)-B_{c}(l,:)\|_{F}, \mathrm{~for~each~}1\leq k\neq l\leq K_{c},
\end{align}
then the clustering error $\hat{f}_{c}=O(\varsigma^{2})$. Setting $\varsigma=\frac{2}{\delta_{c}}\sqrt{\frac{K_{c}}{n_{c,\mathrm{min}}}}\|U_{c,*}-\hat{U}_{c,*}\hat{O}\|_{F}$ makes Eq (\ref{holdcODCNM}) hold for all $1\leq k\neq l\leq K_{c}$. Then we have $\hat{f}_{c}=O(\varsigma^{2})=O(\frac{K_{c}\|U_{c,*}-\hat{U}_{c,*}\hat{O}\|^{2}_{F}}{\delta^{2}_{c}n_{c,\mathrm{min}}})$. By Eq (\ref{UC4ODCNM}), we have
\begin{align*}
\hat{f}_{c}=
O(\frac{\theta^{2}_{c,\mathrm{max}}K_{r}K_{c}\mathrm{max}(
\theta_{c,\mathrm{max}}n_{r},\|\theta_{c}\|_{1})n_{c,\mathrm{max}}\mathrm{log}(n_{r}+n_{c})}{\sigma^{2}_{K_{r}}(P)\theta^{4}_{c,\mathrm{min}}\delta^{2}_{c}m^{2}_{V_{c}}\sigma^{2}_{K_{r}}(\Pi_{r})n_{c,K_{r}}n_{c,\mathrm{min}}}).
\end{align*}
Especially, when $K_{r}=K_{c}=K$, $\delta_{c}=\sqrt{2}$ under $ODCNM_{n_{r},n_{c}}(K_{r},K_{c}, P, \Pi_{r}, \Pi_{c},\Theta_{c})$ by Lemma \ref{RKODCNM}, and $m_{V_{c}}=1$. When $K_{r}=K_{c}=K$, we have
\begin{align*}
\hat{f}_{c}=O(\frac{\theta^{2}_{c,\mathrm{max}}K^{2}\mathrm{max}(
\theta_{c,\mathrm{max}}n_{r},\|\theta_{c}\|_{1})n_{c,\mathrm{max}}\mathrm{log}(n_{r}+n_{c})}{\sigma^{2}_{K}(P)\theta^{4}_{c,\mathrm{min}}\sigma^{2}_{K}(\Pi_{r})n^{2}_{c,\mathrm{min}}}).
\end{align*}
\end{proof}
\subsection{Proof of Corollary \ref{AddConditionsODCNM}}
\begin{proof}
For row nodes, under conditions of Corollary \ref{AddConditionsODCNM}, we have
\begin{align*}
\mathrm{max}_{i_{r}\in[n_{r}]}\|e'_{i_{r}}(\hat{\Pi}_{r}-\Pi_{r}\mathcal{P}_{r})\|_{1}=O(\varpi K_{r}\sqrt{\frac{n_{r}}{K_{r}}})=O(\varpi\sqrt{Kn_{r}}).
\end{align*}
Under  conditions of Corollary \ref{AddConditionsODCNM}, $\kappa(\Omega)=O(1)$ and $\mu\leq C\frac{\theta^{2}_{c,\mathrm{max}}}{\theta^{2}_{c,\mathrm{min}}}\leq C$ for some $C>0$ by Lemma 2 of \cite{qing2021DiDCMM}. Then, by Lemma \ref{rowwiseerrorODCNM}, we have
\begin{align*}
\varpi&=O(\frac{\sqrt{\theta_{c,\mathrm{max}}K_{r}}(\kappa(\Omega)\sqrt{\frac{\mathrm{max}(n_{r},n_{c})\mu}{\mathrm{min}(n_{r},n_{c})}}+\sqrt{\mathrm{log}(n_{r}+n_{c})})}{\theta_{c,\mathrm{min}}\sigma_{K_{r}}(P)\sigma_{K_{r}}(\Pi_{r})\sqrt{n_{c,K_{r}}}})\\
&=O(\frac{\sqrt{\theta_{c,\mathrm{max}}K}(\kappa(\Omega)\sqrt{\frac{\mathrm{max}(n_{r},n_{c})\mu}{\mathrm{min}(n_{r},n_{c})}}+\sqrt{\mathrm{log}(n_{r}+n_{c})})}{\theta_{c,\mathrm{min}}\sigma_{K}(P)\sigma_{K}(\Pi_{r})\sqrt{n_{c,\mathrm{min}}}})\\
&=O(\frac{K^{1.5}\sqrt{\theta_{c,\mathrm{max}}}(\sqrt{\frac{C\mathrm{max}(n_{r},n_{c})}{\mathrm{min}(n_{r},n_{c})}}+\sqrt{\mathrm{log}(n_{r}+n_{c})})}{\theta_{c,\mathrm{min}}\sigma_{K}(P)\sqrt{n_{r}n_{c}}}),
\end{align*}
which gives that
\begin{align*}
\mathrm{max}_{i_{r}\in[n_{r}]}\|e'_{i_{r}}(\hat{\Pi}_{r}-\Pi_{r}\mathcal{P}_{r})\|_{1}=O(\frac{K^{2}\sqrt{\theta_{c,\mathrm{max}}}(\sqrt{\frac{C\mathrm{max}(n_{r},n_{c})}{\mathrm{min}(n_{r},n_{c})}}+\sqrt{\mathrm{log}(n_{r}+n_{c})})}{\theta_{c,\mathrm{min}}\sigma_{K}(P)\sqrt{n_{c}}}).
\end{align*}
The reason that we do not consider the case when $K_{r}<K_{c}$ for row nodes is similar as that of Corollary \ref{AddConditions}, and we omit it here.

For column nodes, under conditions of Corollary \ref{AddConditionsODCNM}, we have
  \begin{align*}
\hat{f}_{c}&=O(\frac{\theta^{2}_{c,\mathrm{max}}K_{r}K_{c}\mathrm{max}(
\theta_{c,\mathrm{max}}n_{r},\|\theta_{c}\|_{1})n_{c,\mathrm{max}}\mathrm{log}(n_{r}+n_{c})}{\sigma^{2}_{K_{r}}(P)\theta^{4}_{c,\mathrm{min}}\delta^{2}_{c}m^{2}_{V_{c}}\sigma^{2}_{K_{r}}(\Pi_{r})n_{c,K_{r}}n_{c,\mathrm{min}}})\\
&=O(\frac{\theta^{2}_{c,\mathrm{max}}K^{2}_{r}K^{2}_{c}\mathrm{max}(
\theta_{c,\mathrm{max}}n_{r},\|\theta_{c}\|_{1})\mathrm{log}(n_{r}+n_{c})}{\sigma^{2}_{K_{r}}(P)\theta^{4}_{c,\mathrm{min}}\delta^{2}_{c}m^{2}_{V_{c}}n_{r}n_{c}}).
  \end{align*}
For the case $K_{r}=K_{c}=K$, we have
\begin{align*}
\hat{f}_{c}&=O(\frac{\theta^{2}_{c,\mathrm{max}}K^{2}\mathrm{max}(
\theta_{c,\mathrm{max}}n_{r},\|\theta_{c}\|_{1})n_{c,\mathrm{max}}\mathrm{log}(n_{r}+n_{c})}{\sigma^{2}_{K}(P)\theta^{4}_{c,\mathrm{min}}\sigma^{2}_{K}(\Pi_{r})n^{2}_{c,\mathrm{min}}})\\
&=O(\frac{\theta^{2}_{c,\mathrm{max}}K^{4}\mathrm{max}(
\theta_{c,\mathrm{max}}n_{r},\|\theta_{c}\|_{1})\mathrm{log}(n_{r}+n_{c})}{\sigma^{2}_{K}(P)\theta^{4}_{c,\mathrm{min}}n_{r}n_{c}}).
\end{align*}
When $n_{r}=O(n), n_{c}=O(n), K_{r}=O(1)$ and $K_{c}=O(1)$, the corollary follows immediately by basic algebra.
\end{proof}
\section{The degee-corrected overlapping and nonoverlapping model}\label{DCOOOOOONM}
Here, we extend ONM by introducing degree heterogeneities for row nodes with overlapping property in the directed network $\mathcal{N}$. Let $\theta_{r}$ be an $n_{r}\times 1$ vector whose $i_{r}$-th entry is the degree heterogeneity of row node $i_{r}$, for $i_{r}\in[n_{r}]$. Let $\Theta_{r}$ be an $n_{r}\times n_{r}$ diagonal matrix whose $i_{r}$-th diagonal element is $\theta_{r}(i_{r})$. The extended model for generating $A$ is as follows:
\begin{align}\label{ONMDCONM}
\Omega:=\Theta_{r}\Pi_{r}P\Pi'_{c},~~~A(i_{r},i_{c})\sim\mathrm{Bernoulli}(\Omega(i_{r},i_{c}))\qquad \mathrm{for~}i_{r}\in[n_{r}],i_{c}\in[n_{c}].
\end{align}
\begin{defin}
Call model (\ref{Pir}), (\ref{Pic}), (\ref{krkc}),(\ref{definP}), (\ref{ONMDCONM}) the Degree-Corrected Overlapping and Nonoverlapping model (DCONM) and denote it by $DCONM_{n_{r},n_{c}}(K_{r},K_{c}, P, \Pi_{r}, \Pi_{c},\Theta_{r})$.
\end{defin}
The following conditions are sufficient for the identifiability of DCONM:
\begin{itemize}
  \item (II1) $\mathrm{rank}(P)=K_{r}, \mathrm{rank}(\Pi_{r})=K_{r}, \mathrm{rank}(\Pi_{c})=K_{c}$, and $P(k,k)=1$ for $k\in[K_{r}]$.
  \item (II2) There is at least one pure row node for each of the $K_{r}$ row communities.
\end{itemize}
For degree-corrected overlapping models, it is popular to require that $P$ has unit-``diagonal'' elements for model identifiability, see the model identifiability requirements on the DCMM model of \cite{MixedSCORE} and the OCCAM model of \cite{OCCAM}.
Follow similar proof as that of Lemma \ref{RK}, we have the following lemma.
\begin{lem}\label{RKDCONM}
Under $DCONM_{n_{r},n_{c}}(K_{r},K_{c}, P, \Pi_{r}, \Pi_{c},\Theta_{r})$, there exist an unique $K_{r}\times K_{r}$ matrix $B_{r}$ and an unique $K_{c}\times K_{r}$ matrix $B_{c}$ such that
\begin{itemize}
\item $U_{r}=\Theta_{r}\Pi_{r}B_{r}$ where $B_{r}=\Theta^{-1}_{r}(\mathcal{I}_{r},\mathcal{I}_{r})U_{r}(\mathcal{I}_{r},:)$.
  \item $U_{c}=\Pi_{c}B_{c}$. Meanwhile, $U_{c}(i_{c},:)=U_{c}(\bar{i}_{c},:)$ when $\ell(i_{c})=\ell(\bar{i}_{c})$ for $i_{c},\bar{i}_{c}\in[n_{c}]$, i.e., $U_{c}$ has $K_{c}$ distinct rows. Furthermore, when $K_{r}=K_{c}=K$, we have $\|B_{c}(k,:)-B_{c}(l,:)\|_{F}=\sqrt{\frac{1}{n_{c,k}}+\frac{1}{n_{c,l}}}$ for all $1\leq k<l\leq K$.
\end{itemize}
\end{lem}
The following proposition guarantees the identifiability of DCONM.
\begin{prop}\label{idDCONM}
	If conditions (II1) and (II2) hold, DCONM is identifiable: For eligible $(P,\Pi_{r}, \Pi_{c},\Theta_{r})$ and $(\check{P},\check{\Pi}_{r}, \check{\Pi}_{c},\check{\Theta}_{r})$, if $\Theta_{r}\Pi_{r}P\Pi'_{c}=\check{\Theta}_{r}\check{\Pi}_{r}\check{P}\check{\Pi}'_{c}$, then  $P=\check{P}, \Pi_{r}=\check{\Pi}_{r}, \Pi_{c}=\check{\Pi}_{c}, \Theta_{r}=\check{\Theta}_{r}$.
\end{prop}
\begin{proof}
By Lemma \ref{RKDCONM}, since $U_{c}=\Pi_{c}B_{c}=\Pi_{c}U_{c}(\mathcal{I}_{c},:)=\check{\Pi}_{c}U_{c}(\mathcal{I}_{c},:)$, we have $\Pi_{c}=\check{\Pi}_{c}$. Since $\Omega(\mathcal{I}_{r},\mathcal{I}_{c})=\Theta_{r}(\mathcal{I}_{r},\mathcal{I}_{r})\Pi_{r}(\mathcal{I}_{r},:)P\Pi'_{c}(\mathcal{I}_{c},:)=\Theta_{r}(\mathcal{I}_{r},\mathcal{I}_{r})P=U_{r}(\mathcal{I}_{r},:)\Lambda U'_{c}(\mathcal{I}_{c},:)$, we have $\Theta_{r}(\mathcal{I}_{r},\mathcal{I}_{r})=\mathrm{diag}(U_{r}(\mathcal{I}_{r},:)\Lambda U'_{c}(\mathcal{I}_{c},:))$ by the condition that $P(k,k)=1$ for $k\in[K_{r}]$. Therefore, we also have $\check{\Theta}_{r}(\mathcal{I}_{r},\mathcal{I}_{r})=\mathrm{diag}(U_{r}(\mathcal{I}_{r},:)\Lambda U'_{c}(\mathcal{I}_{c},:))$, which gives that $\Theta_{r}(\mathcal{I}_{r},\mathcal{I}_{r})=\check{\Theta}_{r}(\mathcal{I}_{r},\mathcal{I}_{r})$. Since $\check{\Theta}_{r}(\mathcal{I}_{r},\mathcal{I}_{r})\check{P}=U_{r}(\mathcal{I}_{r},:)\Lambda U'_{c}(\mathcal{I}_{c},:)=\Theta_{r}(\mathcal{I}_{r},\mathcal{I}_{r})P$, we have $P=\check{P}$. By Lemma \ref{RKDCONM}, since $U_{r}=\Theta_{r}\Pi_{r}\Theta^{-1}_{r}(\mathcal{I}_{r},\mathcal{I}_{r})U_{r}(\mathcal{I}_{r},:)=\check{\Theta}_{r}\check{\Pi}_{r}\check{\Theta}^{-1}_{r}(\mathcal{I}_{r},\mathcal{I}_{r})U_{r}(\mathcal{I}_{r},:)=\check{\Theta}_{r}\check{\Pi}_{r}\Theta^{-1}_{r}(\mathcal{I}_{r},\mathcal{I}_{r})U_{r}(\mathcal{I}_{r},:)$ and $U_{r}(\mathcal{I}_{r},:)$ is an nonsingular matrix, we have $\Theta_{r}\Pi_{r}=\check{\Theta}_{r}\check{\Pi}_{r}$. Since $\|\Pi_{r}(i_{r},:)\|_{1}=\|\check{\Pi}_{r}(i_{r},:)\|_{1}=1$ for $i_{r}\in[n_{r}]$, we have $\Pi_{r}=\check{\Pi}_{r}$ and $\Theta_{r}=\check{\Theta}_{r}$.
\begin{rem}\label{Why}
(The reason that we do not introduce a model as an extension of ONM by considering degree heterogeneities for both row and column nodes) Suppose we propose an extension model (call it nontrivial-extension-of-ONM, and ne-ONM for short) of ONM such that $\mathbb{E}[A]=\Omega=\Theta_{r}\Pi_{r}P\Pi'_{c}\Theta_{c}$. For model identifiability, we see that if ne-ONM is identifiable, the following should holds: when $\Omega=\Theta_{r}\Pi_{r}P\Pi'_{c}\Theta_{c}=\check{\Theta}_{r}\check{\Pi}_{r}\check{P}\check{\Pi}'_{c}\check{\Theta}_{c}$, we have $\Theta_{r}=\check{\Theta}_{r}, \Pi_{r}=\check{\Pi}_{r},P=\check{P}, \Pi_{c}=\check{\Pi}_{c}$ and $\Theta_{c}=\check{\Theta}_{c}$. Now we check the identifiability of ne-ONM. Follow proof of Lemma \ref{idDCONM}, since $\Omega(\mathcal{I}_{r},\mathcal{I}_{c})=\Theta_{r}(\mathcal{I}_{r},\mathcal{I}_{r})\Pi_{r}(\mathcal{I}_{r},:)P\Pi'_{c}(\mathcal{I}_{c},:)\Theta_{c}(\mathcal{I}_{c},\mathcal{I}_{c})=\Theta_{r}(\mathcal{I}_{r},\mathcal{I}_{r})P\Theta_{c}(\mathcal{I}_{c},\mathcal{I}_{c})=U_{r}(\mathcal{I}_{r},:)\Lambda U'_{c}(\mathcal{I}_{c},:)$, we have $\Theta_{r}(\mathcal{I}_{r},\mathcal{I}_{r})P=U_{r}(\mathcal{I}_{r},:)\Lambda U'_{c}(\mathcal{I}_{c},:)\Theta^{-1}_{c}(\mathcal{I}_{c},\mathcal{I}_{c})$. If we assume that $P(k,k)=1$ for $k\in[K_{r}]$, we have $\Theta_{r}(\mathcal{I}_{r},\mathcal{I}_{r})=\mathrm{diag}(U_{r}(\mathcal{I}_{r},:)\Lambda U'_{c}(\mathcal{I}_{c},:)\Theta^{-1}_{c}(\mathcal{I}_{c},\mathcal{I}_{c}))$. Similarly, we have $\check{\Theta}_{r}(\mathcal{I}_{r},\mathcal{I}_{r})=\mathrm{diag}(U_{r}(\mathcal{I}_{r},:)\Lambda U'_{c}(\mathcal{I}_{c},:)\check{\Theta}^{-1}_{c}(\mathcal{I}_{c},\mathcal{I}_{c}))$, and it is impossible to guarantee the uniqueness of $\Theta_{r}(\mathcal{I}_{r},\mathcal{I}_{r})$ such that $\check{\Theta}_{r}(\mathcal{I}_{r},\mathcal{I}_{r})$ unless we further assume that $\Theta_{c}(\mathcal{I}_{c},\mathcal{I}_{c})$ is a fixed matrix. However, when we fix $\Theta_{c}(\mathcal{I}_{c},\mathcal{I}_{c})$ such that ne-ONM is identifiable, ne-ONM is nontrivial due to the fact $\Theta_{c}(\mathcal{I}_{c},\mathcal{I}_{c})$ is fixed. And ne-ONM is trivial only when we set $\Theta_{c}(\mathcal{I}_{c},\mathcal{I}_{c})=I_{K_{c}}$, however, for such ne-ONM when $\Theta_{c}(\mathcal{I}_{c},\mathcal{I}_{c})=I_{K_{c}}$, ne-ONM is DCONM actually. The above analysis proposes the reason that why we do not extend ONM by considering $\Theta_{r}$ and $\Theta_{c}$ simultaneously.
\end{rem}
Follow similar idea as \cite{qing2021DiDCMM}, we can design spectral algorithm with consistent estimation to fit DCONM. Compared with ONM and ODCNM, the identifiability requirement of DCONM on $P$ is too strict such that DCONM only model directed network generated from $P$ with diagonal ``unit'' elements, and this is the reason we do not provide DCONM in the main text and propose further algorithmic study as well as theoretical study for it. .
\end{proof}
\vskip 0.2in
\bibliography{refonm}
\end{document}